\documentclass[11pt,oneside,reqno,english]{amsart}
\usepackage[T1]{fontenc}
\usepackage[utf8]{inputenc}
\usepackage[a4paper]{geometry}
\usepackage{color}
\usepackage{babel}
\usepackage{mathtools}
\usepackage{bm}
\usepackage{amstext}
\usepackage{amsthm}
\usepackage{amssymb}
\usepackage{setspace}
\usepackage[bookmarks=false,
 breaklinks=false,pdfborder={0 0 1},backref=false,colorlinks=true]
 {hyperref}
\hypersetup{allcolors=blue}

\makeatletter
\numberwithin{equation}{section}
\numberwithin{figure}{section}
\theoremstyle{plain}
\newtheorem{thm}{\protect\theoremname}
\theoremstyle{definition}
\newtheorem{defn}[thm]{\protect\definitionname}
\theoremstyle{plain}
\newtheorem{prop}[thm]{\protect\propositionname}
\theoremstyle{plain}
\newtheorem{cor}[thm]{\protect\corollaryname}
\theoremstyle{remark}
\newtheorem{rem}[thm]{\protect\remarkname}
\theoremstyle{definition}
\newtheorem{example}[thm]{\protect\examplename}

\usepackage[foot]{amsaddr}
\renewcommand{\email}[2][]{%
  \ifx\emails\@empty\relax\else{\g@addto@macro\emails{,\space}}\fi%
  \@ifnotempty{#1}{\g@addto@macro\emails{\textrm{(#1)}\space}}%
  \g@addto@macro\emails{#2}%
}

\newtheorem{asm}[thm]{Assumption}

\DeclareMathOperator{\ad}{ad}
\DeclareMathOperator{\spn}{span}

\allowdisplaybreaks

\providecommand{\corollaryname}{Corollary}
\providecommand{\definitionname}{Definition}
\providecommand{\examplename}{Example}
\providecommand{\propositionname}{Proposition}
\providecommand{\remarkname}{Remark}
\providecommand{\theoremname}{Theorem}

\makeatother

\providecommand{\corollaryname}{Corollary}
\providecommand{\definitionname}{Definition}
\providecommand{\examplename}{Example}
\providecommand{\propositionname}{Proposition}
\providecommand{\remarkname}{Remark}
\providecommand{\theoremname}{Theorem}

\begin{document}
\title{Non-autonomous soliton hierarchies}
\author{Maciej Błaszak$^{\dagger}$}
\author{Krzysztof Marciniak$^{\ddagger}$}
\author{Błażej Szablikowski$^{\dagger}$}
\address{$^{\dagger}$Faculty of Physics, Adam Mickiewicz University, Uniwersytetu
Poznańskiego 2, 61-614 Poznań, Poland}
\address{$^{\ddagger}$Department of Science and Technology, Linköping University,
Campus Norrköping, 601 74 Norrköping, Sweden}
\email{blaszakm@amu.edu.pl}
\email{krzma@itn.liu.se}
\email{bszablik@amu.edu.pl}
\begin{abstract}
A formalism of systematic construction of integrable non-autonomous
deformations of soliton hierarchies is presented. The theory is formulated
as an initial value problem (IVP) for an associated Frobenius integrability
condition on a Lie algebra. It is showed that this IVP has a formal
solution and within the framewrok of two particular subalgebras of the hereditary
Lie algebra the explicit forms of this formal solution are presented. Finally, this
formalism is applied to Korteveg-de Vries, dispersive water waves and Ablowitz-Kaup-Newell-Segur
soliton hierarchies. 
\end{abstract}

\maketitle

\section{Introduction}

In this article we present a systematic method of deforming of commuting
hierarchies of \emph{autonomous} evolutionary flows, i.e. systems
of evolutionary PDE's of the form 
\begin{equation}
u_{t_{n}}=K_{n}[u],\qquad n=1,2,\ldots,\label{ah}
\end{equation}
(where $u=u(x)=\left(  u_{1}(x),\ldots,u_{N}(x)\right)  ^{T}$ and where each $K_{n}[u]$
is some vector field depending on $u$ and a finite number of its
$x$-derivatives, but not explicitely on times (i.e. evolution parameters
$t_{i}$) and such that 
\begin{equation}
\left[K_{m},K_{n}\right]=0,\qquad m,n=1,2,\ldots,\label{Fa}
\end{equation}
to the \emph{non-autonomous} hierarchies of evolutionary flows 
\begin{equation}
u_{t_{n}}=\mathbb{K}_{n}\left([u],x,t_{1},\ldots,t_{n}\right)\qquad n=1,2,\ldots,\label{nah}
\end{equation}
that satisfy the Frobenius integrability condition 
\begin{equation}
\frac{\partial\mathbb{K}_{n}}{\partial t_{m}}-\frac{\partial\mathbb{K}_{m}}{\partial t_{n}}+\left[\mathbb{K}_{m},\mathbb{K}_{n}\right]=0,\qquad m,n=1,2,\ldots.\label{Fna}
\end{equation}
In \eqref{ah} and \eqref{nah} $K_{n}$ and $\mathbb{K}_{n}$ are
vector fields, depending on $u$ and a finite number of its $x$-derivatives
on some infinite-dimensional functional manifold and $K_{n}$ do not
depend explicitly on times $t_{i}$. Note that in \eqref{nah} we assume
\emph{triangular} dependence of vector fields $\mathbb{K}_{n}$ on
times $t_{i}$ for $1\leqslant i\leqslant n$, which in result simplifies
the Frobenius condition \eqref{Fna} to 
\begin{equation}
\frac{\partial\mathbb{K}_{n}}{\partial t_{m}}+\left[\mathbb{K}_{m},\mathbb{K}_{n}\right]=0,\qquad m<n.\label{tFc}
\end{equation}
The condition \eqref{Fa} (which is nothing else than 
the Frobenius integrability condition for the autonomous system \eqref{ah})
means that the system \eqref{ah} has a
common, multi-time solution through each initial condition $u(x,0,0,\ldots)=u_{0}(x)$.
Likewise, the Frobenius condition
\eqref{Fna} means that the system \eqref{nah} has a common, multi-time
solution through each initial condition $u(x,0,0,\ldots)=u_{0}(x)$.
If these compatibility conditions are not met it makes no sense to
consider the systems in \eqref{ah} or the systems in \eqref{nah}
as hierarchies; they are simply not compatible.

In order to highlight the main algebraic ingredients of our construction
we first (Section \ref{s2}) formulate our theory in a more general
framework, as an initial-value problem (IVP) for $\mathcal{A}$-valued
functions $\mathbb{K}_{n}$ satisfying \eqref{tFc}, where $\mathcal{A}$
is a non-abelian Lie algebra. In Theorem \ref{lemacik} and
Corollary \ref{pikny} we show that this IVP has a formal solution.
Next, we present solutions of the IVP for particular subalgebras of
the hereditary Lie algebra \cite{fuch,fuch2,Ma1996} (Section \ref{s3}) and
then we apply these results to soliton hierarchies (Section \ref{s4}).
In Section \ref{s4} we also find the zero-curvature representations
for the non-autonomous hierarchies \eqref{nah} from zero-curvature
representations of the corresponding autonomous hierarchies \eqref{ah}.
We illustrate our method on three examples: KdV hierarchy (Section
\ref{s5}), dispersive water wave (DWW) hierarchy in the framework of \cite{AF1,AF2} (Section \ref{s5}) and Ablowitz-Kaup-Newell-Segur
(AKNS) hierarchy \cite{AKNS} (Section \ref{s7}).

We believe that the results presented in this article are important
as the majority of research in the theory of integrable PDE's focuses
on autonomous systems of type \eqref{ah}. To our best knowledge,
the non-autonomous deformations \eqref{nah} of soliton hierarchies
have not been previously studied. The usual approach to non-autonomous
soliton equations is to modify a single chosen soliton equation by
assuming some time-dependence of one or more of its coefficients,
see for example \cite{serkin,mani} or \cite{WX}.

This article was inspired by our results presented in \cite{frob0,frob1,frob2},
where we have considered polynomial in times deformations of autonomous
Liouville integrable finite dimensional systems, i.e. systems of the
form 
\begin{align*}
 & \frac{dx}{dt_{i}}=\pi dh_{i},\qquad h_{i}=h_{i}(x),\qquad i=1,\ldots,n,\\
 & \left\{ h_{i},h_{j}\right\} _{\pi}=0,\qquad i,j=1,\ldots,n,
\end{align*}
on some $2n$-dimensional manifold equipped with a Poisson bivector
$\pi$, to non-autonomous Frobenius integrable systems 
\begin{align*}
 & \frac{dx}{dt_{i}}=\pi dH_{i},\qquad H_{i}=H_{i}(x,t_{1},\ldots,t_{n}),\qquad i=1,\ldots,n,\\
 & \left\{ H_{i},H_{j}\right\} _{\pi}+\frac{\partial H_{i}}{\partial t_{j}}-\frac{\partial H_{j}}{\partial t_{i}}=0,\qquad i,j=1,\ldots,n,
\end{align*}
on the same manifold; here $x$ denotes points on this manifold. Another
inspiration for this work was the article \cite{blasz2023} in which
the author constructed non-autonomous KdV hierarchies from Painlevé
systems that were obtained as deformations of Stäckel separable systems.
This construction, however, was dependent on the dimension $n$ of
the underlying Stäckel systems; increasing $n$ led to a completely
different finite hierarchy. This drawback is not present in the theory
we develop in this article.

\section{Frobenius integrability condition in Lie algebras}

\label{s2}

In this section we present a \emph{formal} solution to the Frobenius
integrability condition \eqref{tFc} formulated as an initial-value
problem (IVP) for finite or infinite sets $\left\{ \mathbb{K}_{0},\mathbb{K}_{1},\ldots\right\} $
of elements $\mathbb{K}_{n}$ that belong to a non-abelian Lie algebra
$\mathcal{A}$ and such that each element $\mathbb{K}_{n}$ depends
on (at most) $n+1$ real evolutionary parameters (times) $t_{i}$,
so that 
\begin{equation}
\mathbb{K}_{n}=\mathbb{K}_{n}(t_{0},\ldots,t_{n}).\label{T}
\end{equation}
Thus, all elements $\mathbb{K}_{n}$ will be some $\mathcal{A}$-valued
functions of a finite number of real parameters $t_{i}$. The word
\emph{formal} means in this context that we do not consider any convergence
issues that may arise in the formulas presented in this section. However,
in the next sections we will apply our theory to soliton hierarchies
and then all the expressions appearing in this section will become
convergent (in fact, finite); the $\mathcal{A}$-valued functions
$\mathbb{K}_{n}$ will then become Frobenius integrable deformations
of soliton hierarchies. 
\begin{defn}
We say that the set (finite or not) $\left\{ \mathbb{K}_{0},\mathbb{K}_{1},\ldots\right\} $
of $\mathcal{A}$-valued functions $\mathbb{K}_{i}=\mathbb{K}_{i}(t_{0},\ldots,t_{i})$
satisfies the Frobenius integrability condition (in a triangular form)
if 
\begin{equation}
\frac{\partial\mathbb{K}_{j}}{\partial t_{i}}+\ad_{\mathbb{K}_{i}}\mathbb{K}_{j}=0,\qquad0\leqslant i<j.\label{Frob}
\end{equation}
\end{defn}

Here and in what follows $\ad_{\mathbb{K}_{i}}$ is the adjoint action
in the Lie algebra $\mathcal{A}$ so that $\ad_{\mathbb{K}_{i}}\mathbb{K}_{j}\equiv\left[\mathbb{K}_{i},\mathbb{K}_{j}\right]$.
Note that the \emph{triangular} dependence of $\mathbb{K}_{n}$ on
$t_{i}$ given in \eqref{T} means that \eqref{Frob} is the actual
complete Frobenius condition 
\[
\frac{\partial\mathbb{K}_{j}}{\partial t_{i}}-\frac{\partial\mathbb{K}_{i}}{\partial t_{j}}+\left[\mathbb{K}_{i},\mathbb{K}_{j}\right]=0,\qquad0\leqslant i<j,
\]
simply because \eqref{T} implies that $\frac{\partial\mathbb{K}_{i}}{\partial t_{j}}=0$
for $i<j$.

We begin our exposition by presenting a well-known fact. 
\begin{prop}
Consider the following linear initial value problem 
\begin{equation}
\frac{d\bm{v}}{dt}+\mathbb{A}\bm{v}=0,\qquad\bm{v}(0)=\bm{v}_{0},\label{0}
\end{equation}
where $\mathbb{A}=\mathbb{A}(t)$ is some time-dependent linear operator
acting in $\mathcal{A}$, $\bm{v}=\bm{v}(t)$ is an $\mathcal{A}$-valued
function and where $\bm{v}_{0}\in\mathcal{A}$. Then the formal solution
of this problem is 
\begin{align*}
\bm{v}(t) & =\bm{v}_{0}-\partial_{t}^{-1}\mathbb{A}\bm{v}_{0}+\partial_{t}^{-1}\mathbb{A}\partial_{t}^{-1}\mathbb{A}\bm{v}_{0}-\partial_{t}^{-1}\mathbb{A}\partial_{t}^{-1}\mathbb{A}\partial_{t}^{-1}\mathbb{A}\bm{v}_{0}+\ldots\\
 & \equiv(1+\partial_{t}^{-1}\mathbb{A})^{-1}\bm{v}_{0},
\end{align*}
where 
\[
\partial_{t}^{-1}=\int_{0}^{t}dt^{\prime}
\]
is the formal linear operator of definite integration with respect
to the time variable $t$ so that $\partial_{t}^{-1}\mathbb{A}$ is
a $t$-dependent operator on  $\mathcal{A}$.
\end{prop}

\begin{proof}
Indeed, given that for any linear operator $\mathbb{B}$ in $\mathcal{A}$
we have, formally, 
\begin{equation}
(1+\mathbb{B})^{-1}\equiv1-\mathbb{B}+\mathbb{B}^{2}-\mathbb{B}^{3}+\ldots=1-\mathbb{B}\,(1+\mathbb{B})^{-1},\label{B}
\end{equation}
then, taking $\mathbb{B}=\partial_{t}^{-1}\mathbb{A}$, we obtain
\begin{align*}
\frac{d\bm{v}(t)}{dt} & =\partial_{t}(1+\partial_{t}^{-1}\mathbb{A})^{-1}\bm{v}_{0}=\partial_{t}(\bm{v}_{0}-\partial_{t}^{-1}\mathbb{A}(1+\partial_{t}^{-1}\mathbb{A})^{-1}\bm{v}_{0})\\
 & =-\mathbb{A}(1+\partial_{t}^{-1}\mathbb{A})^{-1}\bm{v}_{0}=-\mathbb{A}\bm{v}(t).
\end{align*}
\end{proof}
In the case $\mathbb{A}$ does not depend on $t$ we obtain the well
known solution of the IVP \eqref{0} in the form of the exponential
of $\mathbb{A}$: 
\begin{equation}
\frac{d\mathbb{A}}{dt}=0\quad\Rightarrow\quad\bm{v}(t)=(1+\partial_{t}^{-1}\mathbb{A})^{-1}\bm{v}_{0} = \exp(-t\mathbb{A})\bm{v}_{0}.\label{nz}
\end{equation}
We will now generalize this result to the case of a system of linear
time-dependent equations. As notified above, we will work with $\mathcal{A}$-valued
functions of some parameters $t_{i}$, which we will often call \emph{times}. 
\begin{thm}
\label{lemacik}Suppose that $n$ $\mathcal{A}$-valued functions
$\mathbb{K}_{0},\ldots,\mathbb{K}_{n-1}$ satisfy the following conditions:
\begin{subequations} \label{w} 
\begin{align}
\mathbb{K}_{i}=\mathbb{K}_{i}(t_{0},\ldots,t_{i}), & \qquad i=0,1,\ldots,n-1,\label{w1}\\
\frac{\partial\mathbb{K}_{j}}{\partial t_{i}}+\ad_{\mathbb{K}_{i}}\mathbb{K}_{j} & =0,\qquad0\leqslant i<j<n.\label{w2}
\end{align}
\end{subequations} 
Then the initial
value problem 
\begin{subequations} \label{1} 
\begin{align}
\frac{\partial\mathbb{K}}{\partial t_{i}}+\ad_{\mathbb{K}_{i}}\mathbb{K} & =0,\qquad i=0,\ldots,n-1,\label{1a}\\
\mathbb{K}(0,\ldots,0) & =\bar{\mathbb{K}}\in\mathcal{A}.\label{1b}
\end{align}
\end{subequations} 
for the $\mathcal{A}$-valued function
$\mathbb{K}=\mathbb{K}(t_{0},\ldots,t_{n-1})$ has the formal solution 
\begin{equation}
\mathbb{K}(t_{0},\ldots,t_{n-1})=(1+\partial_{t_{n-1}}^{-1}\ad_{\mathbb{K}_{n-1}})^{-1}\cdots(1+\partial_{t_{1}}^{-1}\ad_{\mathbb{K}_{1}})^{-1}(1+\partial_{t_{0}}^{-1}\ad_{\mathbb{K}_{0}})^{-1}\bar{\mathbb{K}},\label{2}
\end{equation}
where 
\[
\partial_{t_{i}}^{-1}f(t_{i})\equiv\int_{0}^{t_{i}}f(t)\,dt.
\] 
\end{thm}

Assuming that $\bar{\mathbb{K}}$ in \eqref{1b} depends on an additional
evolution parameter $t_{n}$ we arrive at the following corollary. 
\begin{cor}
\label{pikny}Suppose that $n$ $\mathcal{A}$-valued functions $\mathbb{K}_{0},\ldots,\mathbb{K}_{n-1}$
satisfy the conditions \eqref{w}. Suppose also that an $\mathcal{A}$-valued
function $\mathbb{K}_{n}\mathbb{=}$ $\mathbb{K}_{n}(t_{0},\ldots,t_{n})$
satisfies the following initial value problem 
\begin{subequations}\label{p}
\begin{align}
\frac{\partial\mathbb{K}_{n}}{\partial t_{i}}+\ad_{\mathbb{K}_{i}}\mathbb{K}_{n} & =0,\qquad0\leqslant i<n,\qquad\label{p1}\\
\mathbb{K}_{n}(\underset{n}{\underbrace{0,\ldots,0}},t_{n}) & =\bar{\mathbb{K}}_{n}(t_{n}),\label{p2}
\end{align}
\end{subequations} 
where $\bar{\mathbb{K}}_{n}(t_{n})$ is an $\mathcal{A}$-valued
function of $t_{n}$. Then, the IVP \eqref{p} has the unique (formal)
solution 
\begin{equation}
\mathbb{K}_{n}(t_{0},\ldots,t_{n})=(1+\partial_{t_{n-1}}^{-1}\ad_{\mathbb{K}_{n-1}})^{-1}\cdots(1+\partial_{t_{1}}^{-1}\ad_{\mathbb{K}_{1}})^{-1}(1+\partial_{t_{0}}^{-1}\ad_{\mathbb{K}_{0}})^{-1}\bar{\mathbb{K}}_{n}(t_{n}).\label{D1}
\end{equation}
\end{cor}

For the proof of Theorem \ref{lemacik}, see Appendix A. Let us now
comment this theorem and the corollary that follows. The conditions
\eqref{w} means that the $\mathcal{A}$-valued functions $\mathbb{K}_{0},\ldots,\mathbb{K}_{n-1}$
satisfy the Frobenius integrability condition \eqref{Frob}. Further
in the article the $\mathcal{A}$-valued functions $\mathbb{K}_{i}$
are vector fields on some finite or infinite-dimensional manifold
$\mathcal{M}$, so that $\mathbb{K}_{i}=\mathbb{K}_{i}(t_{0},\ldots,t_{i},u)$
with $u\in\mathcal{M}$, and the Frobenius condition \eqref{Frob}
in turn implies that the corresponding non-autonomous dynamical systems
\[
\frac{du}{dt_{i}}=\mathbb{K}_{i}(t_{0},\ldots,t_{i},u),\qquad i=0,\ldots,n-1,
\]
posses (at least locally) a common, multi-time solution $u=u(t_{0},\ldots,t_{n-1};u_{0})$
through each point $u_{0}$ of the manifold $\mathcal{M}$. Theorem
\ref{lemacik} yields us then a tool to add one more vector field
$\mathbb{K}$ to the set $\left\{ \mathbb{K}_{0},\ldots,\mathbb{K}_{n-1}\right\} $
so that the set $\left\{ \mathbb{K}_{0},\ldots,\mathbb{K}_{n-1},\mathbb{K}_{n}\equiv\mathbb{K}\right\} $
still satisfies the Frobenius integrability condition \eqref{Frob}.
If now the initial condition $\bar{\mathbb{K}}$ for \eqref{1} depends
additionally on a new evolution parameter $t_{n}$ then the solution
\eqref{D1} in Corollary \ref{pikny} provides us with the set $\left\{ \mathbb{K}_{0},\ldots,\mathbb{K}_{n-1},\mathbb{K}_{n}\equiv\mathbb{K}\right\} $
of non-autonomous vector fields, depending now on one more evolution parameter $t_n$, and such that they satisfy the Frobenius condition
\eqref{Frob} for $0\leqslant i<j\leqslant n$.

Thus, Corollary \ref{pikny} is a useful tool to recursively construct
sets of non-autonomous vector fields, \emph{triangularly} depending
on $t_{i}$ and satisfying the Frobenius condition. This corollary
will be used in the following sections to produce non-autonomous Frobenius
integrable deformations of various soliton hierarchies.

Let us conclude this section with a specification of Theorem~\ref{lemacik}
to the situation when $\mathbb{K}_{i}=\mathbb{K}_{i}(t_{0},\ldots,t_{i-1})$ only (in
such situation we say that neither of $\mathbb{K}_{i}$ depends on its
own evolution parameter $t_{i}$). In this particular situation, in
accordance with \eqref{nz}, 
\[
(1+\partial_{t_{i}}^{-1}\ad_{\mathbb{K}_{i}})^{-1}=\exp(-t_{i}\ad_{\mathbb{K}_{i}}),
\]
and thus we have the following remark. 
\begin{rem}
If $\mathbb{K}_{i}=\mathbb{K}_{i}(t_{0},\ldots,t_{i-1})$ for $i=0,\ldots,n-1$
and $\bar{\mathbb{K}}_{n}(t_{n})=\bar{\mathbb{K}}_{n}$, i.e.~if neither of
$\bar{\mathbb{K}}_{i}$ depends on its own evolution parameter $t_{i}$,
then the solution of the IVP \eqref{p} takes the form
\begin{equation}
\mathbb{K}_{n}(t_{0},\ldots,t_{n-1})=\exp(-t_{n-1}\ad_{\mathbb{K}_{n-1}})\cdots\exp(-t_{1}\ad_{\mathbb{K}_{1}})\exp(-t_{0}\ad_{\mathbb{K}_{0}})\bar{\mathbb{K}}_{n}.\label{D2}
\end{equation}
Naturally, in this case 
\[
\mathbb{K}_{n}(0,\ldots,0)=\bar{\mathbb{K}}_{n}.
\]
\end{rem}

\section{Frobenius integrability in hereditary algebras}

\label{s3}

The formulas \eqref{D1} and \eqref{D2} cannot be of any practical
use until we have some method to compute the expressions on their
right hand sides. To achieve this we will assume that our non-abelian
Lie algebra $\mathcal{A}$ is a semi-product of a commutative
algebra and the Witt algebra (a centerless Virasoro symmetry algebra),
i.e. $\mathcal{A}$ is the so-called hereditary algebra \cite{fuch,fuch2,Ma1996}.
Specifically, we assume that the hereditary algebra $\mathcal{A}$
is generated by the elements $K_{n}\in$ $\mathcal{A}$, $n=1,2,\ldots$,
and $\sigma_{m}\in$ $\mathcal{A}$, $m=-1,0,1,\ldots$, such that
\[
[K_{n},K_{m}]=0,\qquad[\mathcal{\sigma}_{n},K_{m}]=(\alpha m+\rho'-1)K_{n+m},\qquad[\sigma_{n},\sigma_{m}]=\alpha(m-n)\sigma_{n+m},
\]
where $\rho',\alpha\in\mathbb{R}$ and $\alpha\neq0$. This choice
is motivated by the fact that hereditary algebras of soliton hierarchies,
that we will consider in this article, have this structure.
By a simple rescaling of all $\sigma_{n}$ we can always set $\alpha=1$.
Thus, in this article we will consider the hereditary algebra $\mathcal{A}$
with the generators $K_{n}$ and $\sigma_{m}$ satisfying the commutation relations
\begin{align}
[K_{n},K_{m}] & =0,\qquad m,n=1,2,\ldots,\nonumber \\{}
[\mathcal{\sigma}_{n},K_{m}] & =\kappa_{m}K_{n+m},\qquad n=-1,0,1,\ldots,\quad m=1,2,\ldots,\label{A}\\{}
[\sigma_{n},\sigma_{m}] & =(m-n)\sigma_{n+m},\qquad m,n=-1,0,1,\ldots,\nonumber 
\end{align}
(so that $\rho-1\equiv\frac{1}{\alpha}\left(\rho'-1\right)$) and
where we denote $K_{0}\equiv0$. Here and further on we use the notation
\begin{equation}
\kappa_{m}\equiv\rho+m-1.\label{km}
\end{equation}
We can now choose our initial conditions $\bar{\mathbb{K}}_{n}(t_{n})$ in \eqref{p2}
as some very particular deformations (times-dependent linear combinations)
of the above generators of $\mathcal{A}$. It turns out that formulas
\eqref{D1} and \eqref{D2} attain a compact, finite form in
two particular cases: 
\begin{enumerate}
\item $\bar{\mathbb{K}}_{n}(t_{n})$ is a $\mathcal{A}_{-1}$-valued function
for all $n$, 
\item $\bar{\mathbb{K}}_{n}(t_{n})$ is a $\mathcal{A}_{0}$-valued function
for all $n$, 
\end{enumerate}
where 
\[
\mathcal{A}_{-1}:=\spn\left\{ \sigma_{-1},K_{1},K_{2},\ldots\right\} \quad\text{and\ensuremath{\quad}}\mathcal{A}_{0}:=\spn\left\{ \sigma_{0},K_{1},K_{2},\ldots\right\} 
\]
are two particular subalgebras of the hereditary algebra $\mathcal{A}$.
They are exceptional in the sense that for any $n$ the sets $\mathcal{A}_{\varepsilon}^{(n)}\equiv\spn\left\{ \sigma_{i},K_{1},K_{2},\ldots,K_{n}\right\} $,
where $\varepsilon=-1$ or $\varepsilon=0$, are finite dimensional
subalgebras of $\mathcal{A}$. On the level of the algebras $\mathcal{A}_{\varepsilon}$
we can interpret the solutions $\mathbb{K}_{n}$ of IVP \eqref{p}
as leading to construction of new deformed bases of $\mathcal{A}_{\varepsilon}$
satisfying the Frobenius integrability condition \eqref{Frob}.

\subsection{Frobenius integrability in the hereditary subalgebra $\mathcal{A}_{-1}$}

We begin with the first case, i.e. when $\bar{\mathbb{K}}_{n}(t_{n})$ is a $\mathcal{A}_{-1}$-valued function.
\begin{thm}
\label{ii}Consider the IVP \eqref{p} with the initial conditions
\eqref{p2} in the form 
\begin{equation}
\mathbb{K}_{n}(0,\ldots,0,t_{n})=\bar{\mathbb{K}}_{n}(t_{n})\equiv\sigma_{-1}+\sum_{i=1}^{n}\mathbf{a}_{n,i}\left(t_{n}\right)K_{i},\qquad n=0,1,...,\label{wp2}
\end{equation}
where $\mathbf{a}_{n,i}(t_{n})$ are arbitrary differentiable functions  
(thus $\bar{\mathbb{K}}_{n}(t_{n})$ are $\mathcal{A}_{-1}$-valued functions). 
Then, the solution \eqref{D1} of the IVP \eqref{p} is unique and attains the form 
\begin{equation}
\mathbb{K}_{n}=\sigma_{-1}+\sum_{i=1}^{n}\mathbf{u}_{n,i}(t_{0},\ldots,t_{n})K_{i},\label{Rm1}
\end{equation}
where 
\begin{equation}
\mathbf{u}_{n,1}=-\mathbf{c}_{n},\qquad\mathbf{u}_{n,i}=\frac{(-1)^{i}}{[\kappa_{i}]!}\partial_{t_{0}}^{i-1}\mathbf{c}_{n},\qquad i=2,\ldots,n,\label{def1}
\end{equation}
with 
\begin{equation}
\begin{split}\mathbf{c}_{n}(t_{0},\ldots,t_{n}) & \equiv\sum_{m=2}^{n-1}\sum_{r=2}^{m}\sum_{s=1}^{r-1}\frac{(-1)^{r-1}[\kappa_{r}]!}{(r-s-1)!}(\tau_{m-1})^{r-s-1}\left(\partial_{t_{m}}^{-1}\right)^{s}\mathbf{a}_{m,r}(t_{m})\\
 & \quad+\sum_{r=1}^{n}\frac{(-1)^{r}[\kappa_{r}]!}{(r-1)!}(\tau_{n-1})^{r-1}\mathbf{a}_{n,r}(t_{n}),\qquad n\in\mathbb{N}.
\end{split}
\label{cn}
\end{equation}
\end{thm}

Here and in what follows we use the shorthand notations: 
\begin{equation}
\tau_{m} = t_{0}+t_{1}+\ldots+t_{m},\qquad[\kappa_{r}]! = \kappa_{2}\kappa_{3}\cdots\kappa_{r-1}\kappa_{r},\quad r>1, \qquad[\kappa_{1}]! = 1,\label{tm2}
\end{equation}
and 
\[
\partial_{t_{m}}^{-1}f(t_{m})\equiv\int_{0}^{t_{m}}f(t)\,dt.
\]

Notice that, the first term in \eqref{cn} disappears for $\mathbf{c}_{1}$
and $\mathbf{c}_{2}$. Also notice that, after combining \eqref{def1}
with \eqref{cn} and simplifying, one can see that $\mathbf{u}_{n,i}$
are always polynomial in variables $\kappa_{m}$. This means that
there is no issue with division by zero if one of \eqref{km}, for
some particular $\rho$, is equal to zero. 
\begin{proof}
Direct calculation yields 
\[
\mathbf{c}_{n}(0,\ldots,0,t_{n})=-\mathbf{a}_{n,1}(t_{n}).
\]
Moreover, 
\begin{align*}
(\mathbf{c}_{n})_{t_{0}} & =\sum_{m=2}^{n-1}\sum_{r=3}^{m}\sum_{s=1}^{r-2}\frac{(-1)^{r-1}[\kappa_{r}]!}{(r-s-2)!}(\tau_{m-1})^{r-s-2}\left(\partial_{t_{m}}^{-1}\right)^{s}\mathbf{a}_{m,r}(t_{m})\\
 & \quad+\sum_{r=2}^{n}\frac{(-1)^{r}[\kappa_{r}]!}{(r-2)!}(\tau_{n-1})^{r-2}\mathbf{a}_{n,r}(t_{n}),
\end{align*}
so that 
\[
\left(\partial_{t_{0}}\mathbf{c}_{n}\right)(0,\ldots,0,t_{n})=\kappa_{2}\,\mathbf{a}_{n,2}(t_{n}).
\]
Continuing differentiation of $\mathbf{c}_{n}$ with respect to $t_{0}$
we obtain 
\[
\left(\partial_{t_{0}}^{i-1}\mathbf{c}_{n}\right)(0,\ldots,0,t_{n})=(-1)^{i}[\kappa_{i}]!\,\mathbf{a}_{n,i}(t_{n}),\qquad i=2,3,\ldots,n.
\]
That means that $\mathbb{K}_{n}(0,\ldots,0,t_{n})=\bar{\mathbb{K}}_{n}(t_{n})$
and thus \eqref{Rm1} satisfies (for each $n$) the initial conditions
\eqref{wp2}. Further, \eqref{Rm1} satisfies \eqref{p1} if and only
if 
\begin{subequations}\label{wki2} 
\begin{align}
(\mathbf{u}_{n,n})_{t_{j}} & =0,\\
(\mathbf{u}_{n,i})_{t_{j}}+\kappa_{i}\mathbf{u}_{n,i+1} & =0,\qquad j\leqslant i\leqslant n-1,\\
(\mathbf{u}_{n,i})_{t_{j}}+\kappa_{i}\left(\mathbf{u}_{n,i+1}-\mathbf{u}_{j,i+1}\right) & =0,\qquad1\leqslant i\leqslant j-1,
\end{align}
\end{subequations} 
see Appendix \ref{appB} with $\epsilon=-1$.
In Appendix \ref{appC} it is showed that 
\begin{equation}
(\mathbf{c}_{n})_{t_{j}}\equiv(\mathbf{c}_{n})_{t_{0}}-(\mathbf{c}_{j})_{t_{0}},\qquad1\leqslant j\leqslant n-1,\label{cntj}
\end{equation}
and hence 
\[
\partial_{t_{0}}^{i-1}(\mathbf{c}_{n})_{t_{j}}=\partial_{t_{0}}^{i}\mathbf{c}_{n}-\partial_{t_{0}}^{i}\mathbf{c}_{j},\qquad1\leqslant i\leqslant n-1,
\]
while $\partial_{t_{0}}^{i}\mathbf{c}_{j}=0$ for $i\geqslant j$.
These properties and the fact that $[\kappa_{i+1}]!\equiv[\kappa_{i}]!\,\kappa_{i+1}$
imply that the conditions \eqref{wki2} are identically satisfied.
So, all $\mathbb{K}_{n}$ given by \eqref{Rm1} satisfy the IVP \eqref{p}
with \eqref{p2} being of the particular form \eqref{wp2}. Since the conditions \eqref{wki2}, being linear equations, have a unique solution for any initial conditions, the solutions \eqref{Rm1} and \eqref{D1} must -- for the chosen initial conditions -- coincide. 
\end{proof}
The first few non-autonomous vectors \eqref{Rm1} for the general
IVP \eqref{wp2} have the form 
\begin{align*}
\mathbb{K}_{0} & =\sigma_{-1},\\
\mathbb{K}_{1} & =\sigma_{-1}+\mathbf{a}_{1,1}(t_{1})\,K_{1},\\
\mathbb{K}_{2} & =\sigma_{-1}+\Bigl[\mathbf{a}_{2,1}(t_{2})-\kappa_{2}\,\tau_{1}\,\mathbf{a}_{2,2}(t_{2})\Bigr]K_{1}+\mathbf{a}_{2,2}(t_{2})\,K_{2},\\
\mathbb{K}_{3} & =\sigma_{-1}+\left[\kappa_{2}\,\partial_{t_{2}}^{-1}\mathbf{a}_{2,2}(t_{2})+\mathbf{a}_{3,1}(t_{3})-\kappa_{2}\,\tau_{2}\,\mathbf{a}_{3,2}(t_{3})+\frac{1}{2}\kappa_{2}\kappa_{3}\,\tau_{2}^{2}\,\mathbf{a}_{3,3}(t_{3})\right]K_{1}\\
 & \quad+\Bigl[\mathbf{a}_{3,2}(t_{3})-\kappa_{3}\,\tau_{2}\,\mathbf{a}_{3,3}(t_{3})\Bigr]K_{2}+\mathbf{a}_{3,3}(t_{3})\,K_{3},\\
\mathbb{K}_{4} & =\sigma_{-1}+\biggl[\kappa_{2}\,\partial_{t_{2}}^{-1}\mathbf{a}_{2,2}(t_{2})+\kappa_{2}\,\partial_{t_{3}}^{-1}\mathbf{a}_{3,2}(t_{3})-\kappa_{2}\kappa_{3}\,\tau_{2}\,\partial_{t_{3}}^{-1}\mathbf{a}_{3,3}(t_{3})-\kappa_{2}\kappa_{3}\,\partial_{t_{3}}^{-2}\mathbf{a}_{3,3}(t_{3})\\
 & \qquad\qquad+\mathbf{a}_{4,1}(t_{4})-\kappa_{2}\,\tau_{3}\,\mathbf{a}_{4,2}(t_{4})+\frac{1}{2}\kappa_{2}\kappa_{3}\,\tau_{3}^{2}\,\mathbf{a}_{4,3}(t_{4})-\frac{1}{6}\kappa_{2}\kappa_{3}\kappa_{4}\,\tau_{3}^{3}\,\mathbf{a}_{4,4}(t_{4})\biggr]K_{1}\\
 & \quad+\left[\kappa_{3}\,\partial_{t_{3}}^{-1}\mathbf{a}_{3,3}(t_{3})+\mathbf{a}_{4,2}(t_{4})-\kappa_{3}\,\tau_{3}\,\mathbf{a}_{4,3}(t_{4})+\frac{1}{2}\kappa_{3}\kappa_{4}\,\tau_{3}^{2}\,\mathbf{a}_{4,4}(t_{4})\right]K_{2}\\
 & \quad+\Bigl[\mathbf{a}_{4,3}(t_{4})-\kappa_{4}\,\tau_{3}\,\mathbf{a}_{4,4}(t_{4})\Bigr]K_{3}+\mathbf{a}_{4,4}(t_{4})\,K_{4},
\end{align*}
where $\kappa_{m}$ are defined by \eqref{km} and $\tau_{m}$ by
\eqref{tm2}. 
\begin{example}
\label{E2}Suppose that the functions $\mathbf{a}_{n,i}\left(t_{n}\right)$
in the initial conditions \eqref{wp2} are given by the simple choice
$\mathbf{a}_{n,i}\left(t_{n}\right)=0$ for $i=1,\ldots,n-1$, and
$\mathbf{a}_{n,n}\left(t_{n}\right)=t_{n}$. Then the initial conditions
\eqref{wp2} have the form 
\begin{equation}
\bar{\mathbb{K}}_{0}(t_{0})\equiv\sigma_{-1},\qquad\bar{\mathbb{K}}_{n}(t_{n})\equiv\sigma_{-1}+t_{n}\,K_{i},\qquad n=1,2,...,\label{inc1}
\end{equation}
and the solution of the IVP \eqref{p} is given by \eqref{Rm1} and \eqref{def1} with $\mathbf{c}_{n}$ in the form
\begin{equation}
\mathbf{c}_{n}(t_{0},\ldots,t_{n})=\sum_{m=2}^{n-1}\sum_{s=2}^{m}\frac{(-1)^{m-1}[\kappa_{m}]!}{(m-s)!s!}(\tau_{m-1})^{m-s}t_{m}^{s}+\frac{(-1)^{n}[\kappa_{n}]!}{(n-1)!}(\tau_{n-1})^{n-1}t_{n}.\label{def12}
\end{equation}
Then, using \eqref{def1} with \eqref{def12} we obtain the first
$\mathbb{K}_{n}$ in \eqref{Rm1} in the explicit form 
\begin{align*}
\mathbb{K}_{0} & =\sigma_{-1},\\
\mathbb{K}_{1} & =\sigma_{-1}+t_{1}K_{1},\\
\mathbb{K}_{2} & =\sigma_{-1}-\kappa_{2}(t_{0}+t_{1})t_{2}K_{1}+t_{2}K_{2},\\
\mathbb{K}_{3} & =\sigma_{-1}+\frac{1}{2}\kappa_{2}\bigl[t_{2}^{2}+\kappa_{3}(t_{0}+t_{1}+t_{2})^{2}t_{3}\bigr]K_{1}-\kappa_{3}(t_{0}+t_{1}+t_{2})t_{3}K_{2}+t_{3}K_{3},\\
\mathbb{K}_{4} & =\sigma_{-1}+\frac{1}{2}\kappa_{2}\Bigl[t_{2}^{2}-\frac{1}{3}\kappa_{3}(3t_{0}+3t_{1}+3t_{2}+t_{3})t_{3}^{2}-\frac{1}{3}\kappa_{3}\kappa_{4}\left(t_{0}+t_{1}+t_{2}+t_{3}\right)^{3}t_{4}\Bigr]K_{1}\\
 & \quad+\frac{1}{2}\kappa_{3}\bigl[t_{3}^{2}+\kappa_{4}(t_{0}+t_{1}+t_{2}+t_{3})^{2}t_{4}\bigr]K_{2}-\kappa_{4}(t_{0}+t_{1}+t_{2}+t_{3})t_{4}K_{3}+t_{4}K_{4}.
\end{align*}
\end{example}

\subsection{Frobenius integrability in the hereditary subalgebra $\mathcal{A}_{0}$}

The second possibility arises when the initial conditions are given
by $\mathcal{A}_{0}$-valued functions. 
\begin{thm}
\label{i} Consider the IVP \eqref{p} with the initial conditions \eqref{p2}
in the form 
\begin{equation}
\mathbb{K}_{n}(0,\ldots,0,t_{n})=\bar{\mathbb{K}}_{n}(t_{n})\equiv\sigma_{0}+\sum_{i=1}^{n}\mathbf{a}_{n,i}\left(t_{n}\right)K_{i}\text{,\ensuremath{\qquad}}n=0,1,2,...,\label{wp}
\end{equation}
where $\mathbf{a}_{n,i}(t_{n})$ are arbitrary differentiable functions (thus $\bar{\mathbb{K}}_{n}(t_{n})$
is an $\mathcal{A}_{0}$-valued function). Then, the solution \eqref{D1} is unique and
attains the form 
\begin{equation}
\mathbb{K}_{n}=\sigma_{0}+\sum_{i=1}^{n}\mathbf{u}_{n,i}(t_{0},\ldots,t_{n})K_{i},\label{R0}
\end{equation}
where 
\begin{equation}
\mathbf{u}_{n,i}(t_{0},\ldots,t_{n})=\sum_{r=i}^{n-1}\mathbf{c}_{r,i}\left(t_{r}\right)e^{-\kappa_{i}\tau_{r-1}}+\mathbf{a}_{n,i}\left(t_{n}\right)e^{-\kappa_{i}\tau_{n-1}},\label{def}
\end{equation}
$\tau_{m}$ are again given by \eqref{tm2} 
and where $\mathbf{c}_{r,i}(t_{r})$ are functions that satisfy
non-homogeneous linear IVPs
\begin{equation}
\mathbf{c}_{r,i}^{\prime}(t_{r})+\kappa_{i}\mathbf{c}_{r,i}(t_{r})=\kappa_{i}\mathbf{a}_{r,i}(t_{r}),\qquad\mathbf{c}_{r,i}(0)=0.\label{c}
\end{equation}
\end{thm}

Note that the solution of \eqref{c} is 
\begin{equation}
\mathbf{c}_{r,i}(t_{r})=\kappa_{i}\,e^{-\kappa_{i}t_{r}}\int_{0}^{t_{r}}\mathbf{a}_{r,i}(t)e^{\kappa_{i}t}\,dt.\label{c2}
\end{equation}

In particular, for the choice $\mathbf{a}_{r,i}(t_{i})=t_{i}^{m}$,
$m=0,1,\ldots$, 
\begin{equation}
\mathbf{c}_{r,i}(t_{r})=m!\left(\sum_{k=0}^{m}\frac{(-\kappa_{i})^{k-m}}{k!}t_{r}^{k}-e^{-\kappa_{i}t_{r}}\right),\qquad\kappa_{i}\neq0.\label{crsw}
\end{equation}

\begin{proof}
Fix $n\in\mathbb{N}$ and assume that the solution $\mathbb{K}_{r}$
given by \eqref{R0} satisfy the IVP \eqref{p} for all $r<n$. Clearly,
$\mathbb{K}_{n}(0,\ldots,0,t_{n})=\bar{\mathbb{K}}_{n}(t_{n})$ so
\eqref{R0} satisfy (for each $n$) the initial conditions \eqref{wp}.
Further, \eqref{R0} satisfies \eqref{p1} if and only if 
\begin{subequations}\label{ab}
\begin{align}
(\mathbf{u}_{n,i})_{t_{j}}+\kappa_{i}\mathbf{u}_{n,i} & =0,\qquad j+1\leqslant i\leqslant n,\label{a}\\
(\mathbf{u}_{n,i})_{t_{j}}+\kappa_{i}\left(\mathbf{u}_{n,i}-\mathbf{u}_{j,i}\right) & =0,\qquad1\leqslant i\leqslant j.\label{b}
\end{align}
\end{subequations} 
For proof of \eqref{ab} see Appendix \ref{appB}
with $\epsilon=0$. Note that the equations \eqref{ab} are linear
and thus posses solutions for all $t_{j}$. Consider \eqref{def},
then for $j+1\leqslant i\leqslant n$ we have 
\begin{align*}
(\mathbf{u}_{n,i})_{t_{j}} & =-\kappa_{i}\sum_{r=i}^{n-1}\mathbf{c}_{r,i}\left(t_{r}\right)e^{-\kappa_{i}\tau_{r-1}}-\kappa_{i}\mathbf{a}_{n,i}\left(t_{n}\right)e^{-\kappa_{i}\tau_{n-1}}\\
 & \equiv-\kappa_{i}\mathbf{u}_{n,i},
\end{align*}
so that \eqref{a} is identically true, while for $1\leqslant i\leqslant j$
\begin{align*}
(\mathbf{u}_{n,i})_{t_{j}} & =\mathbf{c}_{j,i}^{\prime}\left(t_{j}\right)e^{-\kappa_{i}\tau_{j-1}}-\kappa_{i}\sum_{r=j+1}^{n-1}\mathbf{c}_{r,i}\left(t_{r}\right)e^{-\kappa_{i}\tau_{r-1}}-\kappa_{i}\mathbf{a}_{n,i}\left(t_{n}\right)e^{-\kappa_{i}\tau_{n-1}}
\end{align*}
and thus 
\[
(\mathbf{u}_{n,i})_{t_{j}}+\kappa_{i}\left(\mathbf{u}_{n,i}-\mathbf{u}_{j,i}\right)=\left(\mathbf{c}_{j,i}^{\prime}\left(t_{j}\right)+\kappa_{i}\left[\mathbf{c}_{j,i}\left(t_{j}\right)-\mathbf{a}_{j,i}\left(t_{j}\right)\right]\right)e^{-\kappa_{i}\tau_{j-1}},
\]
so that \eqref{b} is satisfied provided that the differential equations
\eqref{c} hold. As result, all \eqref{R0} satisfy the IVP \eqref{p}
with \eqref{p2} being of the particular form \eqref{wp} and so,  by
uniqueness of solutions of \eqref{ab} for any initial conditions, the solutions \eqref{R0} and \eqref{D1} must -- for the chosen initial
conditions -- coincide. 
\end{proof}
The first few non-autonomous vectors \eqref{R0} for the general IVP
\eqref{wp} have the form 
\begin{align*}
\mathbb{K}_{0} & =\sigma_{0},\\
\mathbb{K}_{1} & =\sigma_{0}+\mathbf{a}_{1,1}(t_{1})e^{-\kappa_{1}\tau_{0}}K_{1},\\
\mathbb{K}_{2} & =\sigma_{0}+\Bigl[\mathbf{c}_{1,1}(t_{1})e^{-\kappa_{1}\tau_{0}}+\mathbf{a}_{2,1}(t_{2})e^{-\kappa_{1}\tau_{1}}\Bigr]K_{1}+\mathbf{a}_{2,2}(t_{2})e^{-\kappa_{2}\tau_{1}}K_{2},\\
\mathbb{K}_{3} & =\sigma_{0}+\Bigl[\mathbf{c}_{1,1}(t_{1})e^{-\kappa_{1}\tau_{0}}+\mathbf{c}_{2,1}(t_{2})e^{-\kappa_{1}\tau_{1}}+\mathbf{a}_{3,1}(t_{3})e^{-\kappa_{1}\tau_{2}}\Bigr]K_{1}\\
 & \quad+\Bigl[e^{-\kappa_{2}\tau_{1}}\mathbf{c}_{2,2}(t_{2})+\mathbf{a}_{3,2}(t_{3})e^{-\kappa_{2}\tau_{2}}\Bigr]K_{2}+\mathbf{a}_{3,3}(t_{3})e^{-\kappa_{3}\tau_{2}}K_{3},\\
\mathbb{K}_{4} & =\sigma_{0}+\Bigl[\mathbf{c}_{1,1}(t_{1})e^{-\kappa_{1}\tau_{0}}+\mathbf{c}_{2,1}(t_{2})e^{-\kappa_{1}\tau_{1}}+\mathbf{c}_{3,1}(t_{3})e^{-\kappa_{1}\tau_{2}}+\mathbf{a}_{4,1}(t_{4})e^{-\kappa_{1}\tau_{3}}\Bigr]K_{1}\\
 & \quad+\Bigl[\mathbf{c}_{2,2}(t_{2})e^{-\kappa_{2}\tau_{1}}+\mathbf{c}_{3,2}(t_{3})e^{-\kappa_{2}\tau_{2}}+\mathbf{a}_{4,2}(t_{4})e^{-\kappa_{2}\tau_{3}}\Bigr]K_{2}\\
 & \quad+\Bigl[\mathbf{c}_{3,3}(t_{3})e^{-\kappa_{3}\tau_{2}}+\mathbf{a}_{4,3}(t_{4})e^{-\kappa_{3}\tau_{3}}\Bigr]K_{3}+\mathbf{a}_{4,4}(t_{4})e^{-\kappa_{4}\tau_{3}}K_{4},
\end{align*}
where $\kappa_{m}$ are defined by \eqref{km}, $\tau_{m}$ by \eqref{tm2}
and $\mathbf{c}_{r,i}(t_{r})$ are given by \eqref{c2}. 
\begin{example}
\label{E0}Suppose that the functions $\mathbf{a}_{n,i}\left(t_{n}\right)$
in the initial conditions \eqref{wp} are given by $\mathbf{a}_{n,i}(t_{n})=0$
for $i=1,\ldots,n-1$ and $\mathbf{a}_{n,n}(t_{n})=t_{n}$, so that 
\begin{equation}
\bar{\mathbb{K}}_{0}(t_{0})\equiv\sigma_{-1},\qquad\bar{\mathbb{K}}_{n}(t_{n})\equiv\sigma_{-1}+t_{n}\,K_{i},\qquad n=1,2,...\,.\label{inc0}
\end{equation}
In this case, by \eqref{c2} or \eqref{crsw} 
\begin{equation}
\mathbf{c}_{n,i}(t_{n})=0,\qquad i=1,\ldots,n-1,\qquad\bm{c}_{n,n}(t_{n})=t_{n}+\frac{1}{\kappa_{n}}\left(e^{-\kappa_{n}t_{n}}-1\right),\qquad\kappa_{n}\neq0.\label{cni}
\end{equation}
Then, if $\kappa_{i}\neq0$, the first $\mathbb{K}_{n}$ in \eqref{R0}
have the form 
\begin{align*}
\mathbb{K}_{0} & =\sigma_{0},\\
\mathbb{K}_{1} & =\sigma_{0}+t_{1}e^{-\kappa_{1}t_{0}}K_{1},\\
\mathbb{K}_{2} & =\sigma_{0}+\Bigl[\Bigl(t_{1}-\frac{1}{\kappa_{1}}\Bigr)e^{-\kappa_{1}t_{0}}+\frac{1}{\kappa_{1}}e^{-\kappa_{1}(t_{0}+t_{1})}\Bigr]K_{1}+t_{2}e^{-\kappa_{2}(t_{0}+t_{1})}K_{2},\\
\mathbb{K}_{3} & =\sigma_{0}+\Bigl[\Bigl(t_{1}-\frac{1}{\kappa_{1}}\Bigr)e^{-\kappa_{1}t_{0}}+\frac{1}{\kappa_{1}}e^{-\kappa_{1}(t_{0}+t_{1})}\Bigr]K_{1}\\
 & \quad+\Bigl[\Bigl(t_{2}-\frac{1}{\kappa_{2}}\Bigr)e^{-\kappa_{2}(t_{0}+t_{1})}+\frac{1}{\kappa_{2}}e^{-\kappa_{2}(t_{0}+t_{1}+t_{2})}\Bigr]K_{2}+t_{3}e^{-\kappa_{3}(t_{0}+t_{1}+t_{2})}K_{3},\\
\mathbb{K}_{4} & =\sigma_{0}+\Bigl[\Bigl(t_{1}-\frac{1}{\kappa_{1}}\Bigr)e^{-\kappa_{1}t_{0}}+\frac{1}{\kappa_{1}}e^{-\kappa_{1}(t_{0}+t_{1})}\Bigr]K_{1}\\
 & \quad+\Bigl[\Bigl(t_{2}-\frac{1}{\kappa_{2}}\Bigr)e^{-\kappa_{2}(t_{0}+t_{1})}+\frac{1}{\kappa_{2}}e^{-\kappa_{2}(t_{0}+t_{1}+t_{2})}\Bigr]K_{2}\\
 & \quad+\Bigl[\Bigl(t_{3}-\frac{1}{\kappa_{3}}\Bigr)e^{-\kappa_{3}(t_{0}+t_{1}+t_{2})}+\frac{1}{\kappa_{3}}e^{-\kappa_{3}(t_{0}+t_{1}+t_{2}+t_{3})}\Bigr]K_{3}+t_{4}e^{-\kappa_{4}(t_{0}+t_{1}+t_{2}+t_{3})}K_{4}.
\end{align*}
\end{example}

\section{Non-autonomous soliton hierarchies and their deformed isospectral
zero-curvature representations}

\label{s4}

From now on we will assume that the algebraic objects like $K$, $\mathbb{K}$,
or $\sigma$ are vector fields on some infinite-dimensional manifold
$\mathcal{M}$ with the corresponding autonomous evolution equations
$u_{t}=K[u]$ and non-autonomous evolution equations $u_{t}=\mathbb{K}\left(x,t,[u]\right)$,
where the square bracket denotes the dependence on $u$ and a finite
number of derivatives of $u$ w.r.t.~$x$ (so $[u]$ denotes jet-coordinates
on $\mathcal{M}$) and where $u=(u_{1}(x),\ldots,u_{N}(x))^{T}$ denotes
points on the manifold $\mathcal{M}$.

Consider thus an infinite hierarchy of mutually commuting autonomous
evolutionary equations on $\mathcal{M}$ of the form 
\begin{equation}
u_{s_{n}}=K_{n}[u],\qquad n=1,2\ldots,\label{h}
\end{equation}
as well as a hierarchy of non-commuting evolutionary equations on
$\mathcal{M}$:
\begin{equation}
u_{\tau_{n}}=\sigma_{n}[u],\qquad n=-1,0,1,\ldots,\label{hs}
\end{equation}
such that the commutation relations \eqref{A} are valid. The members
of the hierarchy \eqref{hs} are called master symmetries for \eqref{h}.

In this section we obtain -- under Assumption \ref{A1} and Assumption
\ref{A2} -- the Frobenius integrable non-autonomous hierarchies
$u_{t_{n}}=\mathbb{K}_{n}[u]$, where $\mathbb{K}_{n}$ is of the
form \eqref{Rm1} or \eqref{R0}, from an appropriate deformation
of an isospectral zero-curvature representation of \eqref{h} by a
non-standard (see Remark \ref{non-standard}) isospectral zero-curve
representation of \eqref{hs}.

\begin{asm} \label{A1} Suppose that the commuting hierarchy \eqref{h}
can be obtained from the isospectral linear problem 
\begin{equation}
\begin{cases}
\Psi_{x}=L\Psi,\\
\Psi_{s_{i}}=U_{i}\Psi, & i=1,2,\ldots,
\end{cases}\label{4.1}
\end{equation}
where $L=L(\lambda,u)$, $U_{i}=U_{i}(\lambda,[u])$ are some $2\times2$
matrices depending on $[u]$ and the auxiliary variable $\lambda$,
s.t.~$\lambda_{s_{i}}=0$ for all $i$. \end{asm}

The subscript $s_{i}$ denotes the total derivative with
respect to the evolution parameter $s_{i}$. The compatibility condition,
that is the condition for existence of a common multi-time solution
$\Psi(x,s_{1},s_{2},\ldots)$, for the problem \eqref{4.1} is 
\begin{subequations}
\label{4.2ab} 
\begin{align}
(\Psi_{x})_{s_{i}} & =(\Psi_{s_{i}})_{x},\qquad i=1,2,\ldots,\label{4.2a}\\
(\Psi_{s_{i}})_{s_{j}} & =(\Psi_{s_{j}})_{s_{i}},\qquad i,j=1,2\ldots\,.\label{4.2b}
\end{align}
\end{subequations}

The condition \eqref{4.2a} is equivalent to 
\begin{equation}
L_{s_{i}}=\left[U_{i},L\right]+\left(U_{i}\right)_{x}\equiv L^{\prime}\left[K_{i}\right],\qquad i=1,2,\ldots\,.\label{4.3a}
\end{equation}
Throughout the whole article $\Omega^{\prime}\left[K\right]$ denotes
the directional derivative of the tensor field $\Omega$ along the
vector field $K$ on $\mathcal{M}$. The identity in \eqref{4.3a}
is the consequence of Assumption \ref{A1}, while the condition \eqref{4.2b}
is equivalent to 
\begin{equation}
\left(U_{i}\right)_{s_{j}}-\left(U_{j}\right)_{s_{i}}+[U_{i},U_{j}]=0,\qquad i,j=1,2,\ldots.\label{4.3b}
\end{equation}
Thus, Assumption \ref{A1} means that \eqref{4.3a} is equivalent
with the corresponding equation $u_{s_{i}}=K_{i}[u]$ in \eqref{h},
i.e. \eqref{4.3a} is an isospectral zero-curvature representation
for \eqref{h}. It also means that $\Psi_{s_{i}}=\mathcal{L}_{K_{i}}\Psi$
where $\mathcal{L}$ is the Lie derivative on $\mathcal{M}$. Then,
the equation \eqref{4.3b} guarantees that all $K_{i}$ commute, since
\[
(\Psi_{s_{i}})_{s_{j}}-(\Psi_{s_{j}})_{s_{i}}=\mathcal{L}_{K_{j}}\mathcal{L}_{K_{i}}\Psi-\mathcal{L}_{K_{i}}\mathcal{L}_{K_{j}}\Psi=\mathcal{L}_{[K_{j},K_{i}]}\Psi=\Psi^{\prime}\bigl[[K_{j},K_{i}]\bigr]=0.
\]
Note also that \eqref{4.3b} can be written as 
\begin{equation}
U_{i}^{\prime}[K_{j}]-U_{j}^{\prime}[K_{i}]+[U_{i},U_{j}]=0,\qquad i,j=1,2,\ldots\,.\label{ZCU}
\end{equation}

\begin{asm} \label{A2} Suppose also that the (non-commuting) hierarchy
\eqref{hs} of master symmetries can be obtained from the following
\emph{deformed} linear isospectral problem 
\begin{equation}
\begin{cases}
\Psi_{x}=L\Psi,\\
\Psi_{\tau_{i}}=V_{i}\Psi-\lambda^{i+1}\Psi_{\lambda}, & i=-1,0,1\ldots,
\end{cases}\label{4.4}
\end{equation}
where $L=L(\lambda,u)$ is the same $L$ as in \eqref{4.1} while
$V_{i}=V_{i}(\lambda,[u])$ are some matrices depending on $[u]$
and $\lambda$ such that $\lambda_{\tau_{i}}=0$. \end{asm}

In \eqref{4.4} $\Psi_{\lambda}\equiv\frac{\partial\Psi}{\partial\lambda}$.
Obviously, we cannot expect that \eqref{4.4} posses a common multi-time
solution $\Psi(x,\tau_{-1},\tau_{0},\tau_{1},\ldots)$. Instead, the
Assumption \ref{A2} means that \eqref{4.4} has, for each $i$, a
solution $\Psi(x,\tau_{i})$ so that 
\[
(\Psi_{x})_{\tau_{i}}=(\Psi_{\tau_{i}})_{x},\qquad i=-1,0,\ldots,
\]
which is equivalent to 
\begin{equation}
L_{\tau_{i}}=\left[V_{i},L\right]+\left(V_{i}\right)_{x}-\lambda^{i+1}L_{\lambda}\equiv L^{\prime}\left[\sigma_{i}\right],\qquad i=-1,0,1,\ldots,\label{4.6}
\end{equation}
and the identity in \eqref{4.6} is the consequence of Assumption
\ref{A2}. This assumption means thus that each equation in \eqref{4.6}
is equivalent with the corresponding equation $u_{\tau_{i}}=\sigma_{i}[u]$
in \eqref{hs}. Since the fields $\sigma_{i}$ in \eqref{hs} do not
commute we clearly cannot expect that $(\Psi_{\tau_{i}})_{\tau_{j}}=(\Psi_{\tau_{j}})_{\tau_{i}}$.
Instead we have 
\[
(\Psi_{\tau_{i}})_{\tau_{j}}-(\Psi_{\tau_{j}})_{\tau_{i}}=\Psi^{\prime}\bigl[[\sigma_{j},\sigma_{i}]\bigr]=(i-j)\Psi^{\prime}[\sigma_{i+j}]=(i-j)\Psi_{\tau_{i+j}},\qquad i,j=-1,0,1,\ldots,
\]
which is equivalent to 
\[
(V_{i})_{\tau_{j}}-(V_{j})_{\tau_{i}}+[V_{i},V_{j}]+\lambda^{j+1}(V_{i})_{\lambda}-\lambda^{i+1}(V_{j})_{\lambda}=(i-j)V_{i+j},\qquad i,j=-1,0,\ldots,
\]
that is to 
\begin{equation}
V_{i}^{\prime}[\sigma_{j}]-V_{j}^{\prime}[\sigma_{i}]+[V_{i},V_{j}]+\lambda^{j+1}(V_{i})_{\lambda}-\lambda^{i+1}(V_{j})_{\lambda}=(i-j)V_{i+j},\qquad i,j=-1,0,\ldots\,.\label{numer}
\end{equation}

\begin{rem}
\label{non-standard} Usually in literature (see for example \cite{Ma1996})
one constructs a zero-curvature representation for \eqref{hs}, from
the \emph{non-isospectral} problem 
\[
\begin{cases}
\Psi_{x}=L\Psi,\\
\Psi_{\tau_{i}}=V_{i}\Psi, & i=1,2,\ldots,
\end{cases}
\]
with $\lambda_{\tau_{i}}=\lambda^{i+1}$. The isospectral problem
\eqref{4.4} is however equivalent (in the sense that it leads to
the same zero-curvature equations \eqref{numer}) with the above isospectral
problem, while being better adapted to our needs. 
\end{rem}

We will now construct an isospectral zero-curvature representation
of the hierarchies 
\[
u_{t_{n}}=\mathbb{K}_{n}[u],
\]
with $\mathbb{K}_{n}$ given in \eqref{Rm1} or in \eqref{R0} by
combining the isospectral problems \eqref{4.1} and \eqref{4.4}.
Consider thus the \emph{deformed} isospectral linear problem 
\begin{equation}
\begin{cases}
\Psi_{x}=L\Psi,\\
\Psi_{t_{n}}=W_{n}\Psi-\lambda^{\varepsilon+1}\Psi_{\lambda}, & n=1,2,\ldots,
\end{cases}\label{4.8}
\end{equation}
with $\varepsilon=-1$ or $\varepsilon=0$ and where $\lambda_{t_{n}}=0$,
with $W_{n}=W_{n}(\lambda,[u])$ defined as 
\begin{equation}
W_{n}=V_{\varepsilon}+\sum_{i=1}^{n}\mathbf{v}_{n,i}(t_{0},\ldots,t_{n})U_{i},\label{Wn}
\end{equation}
where $L,V_{\varepsilon},U_{i}$ are given as above in this section
and where $\mathbf{v}_{n,i}$ are so far undetermined functions of
evolution parameters $t_{0},\ldots,t_{n}$. 
\begin{thm}
The compatibility condition $(\Psi_{x})_{t_{n}}=(\Psi_{t_{n}})_{x}$
for \eqref{4.8} has the form 
\begin{equation}
L_{t_{n}}=\left[W_{n},L\right]+\left(W_{n}\right)_{x}-\lambda^{\varepsilon+1}L_{\lambda}\equiv L^{\prime}\left[u_{t_{n}}\right],\qquad n=1,2,\ldots,\label{LaxKn}
\end{equation}
where 
\begin{equation}
u_{t_{n}}=\mathbb{K}_{n}\equiv\sigma_{\varepsilon}+\sum_{i=1}^{n}\mathbf{v}_{n,i}(t_{0},\ldots,t_{n})K_{i}[u],\qquad n=1,2,\ldots\,.\label{hd}
\end{equation}
\end{thm}

The identity in \eqref{LaxKn} means that \eqref{LaxKn} is equivalent
with the corresponding equation $u_{t_{n}}=\mathbb{K}_{n}[u]$ in
\eqref{hd}. Note that so far the vector fields $\mathbb{K}_{n}[u]$ in
\eqref{hd} have nothing in common with $\mathbb{K}_{n}$ in  \eqref{Rm1} or \eqref{R0}.
\begin{proof}
Due to the form of $W_{n}$ we have 
\begin{align*}
\left[W_{n},L\right]+\left(W_{n}\right)_{x}-\lambda^{\varepsilon+1}L_{\lambda} & =\left[V_{\varepsilon},L\right]+\left(V_{\varepsilon}\right)_{x}-\lambda^{\varepsilon+1}L_{\lambda}+\sum_{i=1}^{n}\mathbf{v}_{n,i}\left(\left[U_{i},L\right]+\left(U_{i}\right)_{x}\right)\\
 & =L^{\prime}\left[\sigma_{\varepsilon}\right]+\sum_{i=1}^{n}\mathbf{v}_{n,i}L^{\prime}\left[K_{i}\right]\equiv L^{\prime}\left[\mathbb{K}_{n}\right].
\end{align*}
\end{proof}
So far the functions $\mathbf{v}_{n,i}$ are undetermined. Let us
now demand that the compatibility conditions 
\begin{equation}
(\Psi_{t_{m}})_{t_{n}}=(\Psi_{t_{n}})_{t_{m}},\qquad m,n=1,2,\ldots,\label{CC}
\end{equation}
for \eqref{4.8} are satisfied. These conditions are equivalent with
the Frobenius integrability conditions \eqref{Frob}, since 
\begin{align*}
 & (\Psi_{t_{m}})_{t_{n}}-(\Psi_{t_{n}})_{t_{m}}=\left(\Psi^{\prime}[\mathbb{K}_{m}]\right)_{t_{n}}-\left(\Psi^{\prime}[\mathbb{K}_{n}]\right)_{t_{m}}\\
 & =\Psi^{\prime\prime}[\mathbb{K}_{n};\mathbb{K}_{m}]+\Psi^{\prime}\left[\frac{\partial\mathbb{K}_{m}}{\partial t_{n}}+\mathbb{K}_{m}^{\prime}[\mathbb{K}_{n}]\right]-\Psi^{\prime\prime}[\mathbb{K}_{m};\mathbb{K}_{n}]-\Psi^{\prime}\left[\frac{\partial\mathbb{K}_{n}}{\partial t_{m}}+\mathbb{K}_{n}^{\prime}[\mathbb{K}_{m}]\right]\\
 & =\Psi^{\prime}\left[\frac{\partial\mathbb{K}_{m}}{\partial t_{n}}-\frac{\partial\mathbb{K}_{n}}{\partial t_{m}}+[\mathbb{K}_{n},\mathbb{K}_{m}]\right],\qquad m,n=1,2,\ldots,
\end{align*}
where $\Psi^{\prime\prime}[\mathbb{K}_{n};\mathbb{K}_{m}]=\Psi^{\prime\prime}[\mathbb{K}_{m};\mathbb{K}_{n}]$
is the second directional derivative. On the other hand \eqref{CC}
hold if and only if 
\begin{equation}
(W_{n})_{t_{m}}-(W_{m})_{t_{n}}+[W_{n},W_{m}]+\lambda^{\varepsilon+1}(W_{n})_{\lambda}-\lambda^{\varepsilon+1}(W_{m})_{\lambda}=0.\label{ZC}
\end{equation}
Let us thus investigate the conditions under which \eqref{ZC} hold.
Assuming $1\leqslant m<n,$ we have that $(\mathbf{v}_{m,i})_{t_{n}}=0$
and then the zero-curvature relations \eqref{ZC} reduce to 
\begin{equation}
\frac{\partial W_{n}}{\partial t_{m}}+W_{n}^{\prime}[\mathbb{K}_{m}]-W_{m}^{\prime}[\mathbb{K}_{n}]+[W_{n},W_{m}]+\lambda^{\varepsilon+1}(W_{n})_{\lambda}-\lambda^{\varepsilon+1}(W_{m})_{\lambda}=0\text{,}\qquad1\leqslant m<n.\label{ZC1}
\end{equation}

We can now prove the following theorem. 
\begin{thm}
\label{vu}The zero curvature conditions \eqref{ZC1} are equivalent
with the set of equations on the functions $\mathbf{v}_{n,i}$ that
is exactly the same as the set of equations \eqref{wki2} and \eqref{ab},
respectively for $\varepsilon=-1$ and $\varepsilon=0$, on the functions
$\mathbf{u}_{n,i}$.
\end{thm}

The proof of this theorem can be found in Appendix \ref{appD}. This
theorem means that the functions $\mathbf{v}_{n,i}$ and $\mathbf{u}_{n,i}$
pairwise coincide so that the deformed isospectral problem \eqref{4.8}
leads exactly to the Frobenius integrable hierarchy $u_{t_{n}}=\mathbb{K}_{n}[u]$
with $\mathbb{K}_{n}$ given by \eqref{Rm1} (for $\varepsilon=-1$)
or by \eqref{R0} (for $\varepsilon=0$).  Hence,
\begin{equation}
W_{n}\equiv V_{\varepsilon}+\sum_{i=1}^{n}\mathbf{u}_{n,i}(t_{0},\ldots,t_{n})U_{i}.\label{Wnu}
\end{equation}

\begin{rem}
Thus, we obtain the same non-autonomous Frobenius integrable hierarchies of PDE's starting from the deformed spectral
problem \eqref{4.8} and starting from the non-autonomous deformations
in the case of the subalgebras $\mathcal{A}_{\varepsilon}$ of the
hereditary algebra \eqref{A}, that we consider in Section \ref{s3}.
\end{rem}

\subsection{Hamiltonian structure of non-autonomous soliton hierarchies}

\label{HS}

We will now focus on soliton hierarchies. Suppose we have an infinite
hierarchy of vector fields $K_{n}$ on $\mathcal{M}$ that are bi-Hamiltonian
with respect to two compatible Poisson structures $\pi_{0}$ and $\pi_{1}$
\[
K_{n}=\pi_{0}\delta H_{n}=\pi_{1}\delta H_{n-1},\qquad n=1,2,\ldots,
\]
with $\pi_{0}$ being invertible. Then the operator $N=\pi_{1}\pi_{0}^{-1}$
is an operator with the vanishing Nijenhuis torsion so that for any
vector field $K$ on $\mathcal{M}$ we have 
\begin{equation}
\mathcal{L}_{NK}N=N\mathcal{L}_{K}N.\label{her}
\end{equation}
Any operator satisfying \eqref{her} is called a hereditary operator.
Then, 
\begin{equation}
K_{n}\equiv N^{n-1}K_{1},\qquad n=2,3,\ldots,\label{rK}
\end{equation}
and by the hereditary property \eqref{her} 
\[
\mathcal{L}_{K_{n}}N=0,\qquad n=1,2,\ldots,
\]
and 
\[
[K_{n},K_{m}]=0,\qquad n,m=1,2,\ldots\,.
\]

Define now the infinite sequence of $1$-forms 
\begin{equation}
\gamma_{n}\equiv\delta H_{n}=\bigl(N^{\dagger}\bigr)^{i}\gamma_{0},\qquad n=0,1,\ldots\,,\label{1f}
\end{equation}
where $\gamma_{0}=\delta H_{0}$ and where $N^{\dagger}=\pi_{0}^{-1}\pi_{1}$.
By the same hereditary property of $N$ they are all closed and thus
there exists an infinite sequence of functionals $H_{n}=\int h_{n}dx$
such that $\gamma_{n}=\delta H_{n}$. Define now the infinite sequence
of Poisson operators 
\[
\pi_{k}=N^{k}\pi_{0},\qquad k=2,3,\ldots\,.
\]
The operators $\pi_{k}$ are pairwise compatible and usually non-local.
Then, it follows that the field $K_{n}$ is $(n+1)$-Hamiltonian 
\[
K_{n}=\pi_{0}\delta H_{n}=\pi_{1}\delta H_{n-1}=\ldots=\pi_{n}\delta H_{0}\text{,\ensuremath{\qquad}}n=1,2,\ldots\,.
\]
Consider also a scaling vector field $\sigma_{0}$ such that 
\[
\begin{cases}
\mathcal{L}_{\sigma_{0}}K_{1}=\rho K_{1}, & \rho\in\mathbf{R,}\\
\mathcal{L}_{\sigma_{0}}N=N,
\end{cases}
\]
and define the infinite sequences of vector fields $\sigma_{n}$ on
$\mathcal{M}$ through 
\begin{equation}
\sigma_{n}=N^{n}\mathcal{\sigma}_{0},\qquad n=-1,0,1,\ldots\,.\label{rS}
\end{equation}
Then it can be shown, using the hereditary property \eqref{her},
that the vector fields \eqref{rK} and \eqref{rS} satisfy the commutation
relations \eqref{A}.

Finally, let us assume that there exist a vector field $\sigma_{-1}$
such that $\sigma_{0}=N\sigma_{-1}$ and such that it is Hamiltonian
with respect to $\pi_{0}$ 
\[
\mathcal{\sigma}_{-1}=\pi_{0}\delta F.
\]
Then, all $\sigma_{n}$ are Hamiltonian with respect to the Poisson
operator $\pi_{n+1}$, that is
\begin{equation}
\sigma_{n}=\pi_{n+1}\delta F,\qquad n=1,2,...,\label{sH}
\end{equation}
and it immediately follows that every non-autonomous vector field
$\mathbb{K}_{n}$ in \eqref{Rm1} or in \eqref{R0} is also Hamiltonian (but not bi-Hamiltonian),
as 
\begin{align*}
\mathbb{K}_{n} & \equiv\sigma_{\varepsilon}+\sum_{i=1}^{n}\mathbf{u}_{n,i}(t_{0},\ldots,t_{n})K_{i}[u]=\\
 & =\pi_{\varepsilon+1}\delta\left(F+\sum_{i=1}^{n}\mathbf{u}_{n,i}(t_{0},\ldots,t_{n})H_{i-\varepsilon-1}[u]\right).
\end{align*}

\section{Non-autonomous KdV hierarchy}

\label{s5}

In this and following sections we apply our theory, developed above,
to three well known soliton hierarchies: Korteveg-de Vries, dispersive
water waves and Ablowitz-Kaup-Newell-Segur. The all have the structure exactly as described in subsection \ref{HS}.

\subsection{KdV hierarchy}

As a first illustration of our theory, consider the KdV hierarchy.
The KdV equation 
\[
u_{t}=\frac{1}{4}u_{xxx}+\frac{3}{2}uu_{x}
\]
is a member of the bi-Hamiltonian chain of nonlinear PDE's 
\begin{equation}
u_{t_{i}}=K_{i}[u]=\pi_{0}\delta H_{i}=\pi_{1}\delta H_{i-1},\qquad i=1,2,...,\label{h0}
\end{equation}
with two compatible Poisson operators 
\[
\pi_{0}=\partial_{x},\qquad\pi_{1}=\frac{1}{4}\partial_{x}^{3}+u\partial_{x}+\frac{1}{2}u_{x}.
\]
The hierarchy \eqref{h0} is \emph{autonomous} in the sense that none
of the vector fields $K_{i}[u]$ of the hierarchy depends explicitly on the evolution
parameters $t_{j}$. The KdV hierarchy \eqref{h0} can be generated
by the recursion operator and its adjoint 
\[
N\equiv\pi_{1}\pi_{0}^{-1}=\frac{1}{4}\partial_{x}^{2}+u+\frac{1}{2}u_{x}\partial_{x}^{-1},\qquad N^{\dagger}=\frac{1}{4}\partial_{x}^{2}+u-\frac{1}{2}\partial_{x}^{-1}u_{x},
\]
in the sense that \eqref{rK} and \eqref{1f} are valid. In particular,
we find that the first vector fields $K_{n}$ have the form 
\[
\begin{split}K_{1} & =u_{x},\qquad K_{2}=\frac{1}{4}u_{xxx}+\frac{3}{2}uu_{x},\\
K_{3} & =\frac{1}{16}u_{5x}+\frac{5}{8}uu_{3x}+\frac{5}{4}u_{x}u_{xx}+\frac{15}{8}u^{2}u_{x},\\
K_{4} & =\frac{1}{64}u_{7x}+\frac{7}{32}uu_{5x}+\frac{21}{32}u_{x}u_{4x}+\frac{35}{32}u_{xx}u_{3x}+\frac{35}{32}u_{x}^{3}+\frac{35}{8}uu_{x}u_{xx}+\frac{35}{32}u^{2}u_{3x}+\frac{35}{16}u^{3}u_{x},
\end{split}
\]
the first conserved one-forms (cosymmetries) $\gamma_{n}\equiv\delta H_{n}$
are 
\begin{align*}
\gamma_{0} & =2,\qquad\gamma_{1}=u,\qquad\gamma_{2}=\frac{1}{4}u_{xx}+\frac{3}{4}u^{2},\\
\gamma_{3} & =\frac{1}{16}u_{4x}+\frac{5}{8}uu_{xx}+\frac{5}{16}u_{x}^{2}+\frac{5}{8}u^{3},\\
\gamma_{4} & =\frac{1}{64}u_{6x}+\frac{7}{32}uu_{4x}+\frac{7}{16}u_{x}u_{3x}+\frac{21}{64}u_{xx}^{2}+\frac{35}{32}u^{2}u_{xx}+\frac{35}{32}uu_{x}^{2}+\frac{35}{64}u^{4},
\end{align*}
while the first Hamiltonian densities $h_{n}$ of the conserved functionals
$H_{n}=\int h_{n}\,dx$ are 
\[
\begin{split}h_{0} & =2u,\qquad h_{1}=\frac{1}{2}u^{2},\qquad h_{2}=-\frac{1}{8}u_{x}^{2}+\frac{1}{4}u^{3},\\
h_{3} & =\frac{1}{32}u_{xx}^{2}+\frac{5}{32}u^{2}u_{xx}+\frac{5}{32}u^{4},\\
h_{4} & =-\frac{1}{128}u_{3x}^{2}+\frac{7}{64}uu_{xx}^{2}-\frac{35}{64}u^{2}u_{x}^{2}+\frac{7}{64}u^{5}.
\end{split}
\]
With the KdV hierarchy \eqref{h0} one can also relate the hierarchy
of its master symmetries \eqref{rS} with the first few $\sigma_{n}$
of the form 
\begin{align*}
\sigma_{-1} & =1,\qquad\sigma_{0}=u+\frac{1}{2}xu_{x},\\
\sigma_{1} & =\frac{1}{2}u_{xx}+\frac{1}{8}xu_{3x}+u^{2}+\frac{1}{2}xuu_{x}+\frac{1}{4}u_{x}\partial_{x}^{-1}u.
\end{align*}
The master symmetries $\sigma_{n}$ are in general non-local. They
are Hamiltonian due to \eqref{sH} with the conserved functional 
\[
F=\int xu\,dx.
\]
The symmetries $K_{i}$ and master symmetries $\sigma_{j}$ of the
KdV equation generate the hereditary algebra \eqref{A} with 
\begin{equation}
\rho=\frac{1}{2}\qquad\text{so that}\qquad\kappa_{n}=n-\frac{1}{2}.\label{kappa1}
\end{equation}

The matrix Lax representation \eqref{4.3a} for the KdV hierarchy
is given by 
\begin{equation}
L=\begin{pmatrix}0 & 1\\
\lambda-u & 0
\end{pmatrix},\label{LKdV}
\end{equation}
and 
\[
U_{n}=\frac{1}{2}\sum_{i=0}^{n-1}\begin{pmatrix}-\frac{1}{2}\left(\gamma_{i}\right)_{x} & \gamma_{i}\\
(\lambda-u)\gamma_{i}-\frac{1}{2}\left(\gamma_{i}\right)_{xx} & \frac{1}{2}\left(\gamma_{i}\right)_{x}
\end{pmatrix}\lambda^{n-i-1},\qquad n=1,2,\ldots\,.
\]
In particular, $U_{1}=L$, 
\[
U_{2}=\begin{pmatrix}-\frac{1}{4}u_{x} & \lambda+\frac{1}{2}u\\
\lambda^{2}-\frac{1}{2}u\lambda-\frac{1}{2}u^{2}-\frac{1}{4}u_{xx} & \frac{1}{4}u_{x}
\end{pmatrix}
\]
and 
\[
U_{3}=\begin{pmatrix}-\frac{1}{4}u_{x}\lambda-\frac{1}{16}(u_{3x}+6uu_{x}) & \lambda^{2}+\frac{1}{2}u\lambda+\frac{1}{8}(u_{xx}+3u^{2})\\
\lambda^{3}-\frac{1}{2}u\lambda^{2}-\frac{1}{8}(u_{xx}+u^{2})\lambda-\frac{1}{16}u_{4x}+\frac{1}{2}uu_{xx}+\frac{3}{8}u_{x}^{2}+\frac{3}{8}u^{3} & \frac{1}{4}u_{x}\lambda+\frac{1}{16}(u_{3x}+6uu_{x})
\end{pmatrix}.
\]
The Lax formulation for the hierarchy of the KdV master symmetries
$\sigma_{n}$ is given by \eqref{4.6} with $V_{n}$ of the form 
\[
V_{n}=\frac{1}{2}\sum_{i=-1}^{n-1}\begin{pmatrix}-\frac{1}{2}\sigma_{i} & \partial_{x}^{-1}\sigma_{i}\\
(\lambda-u)\partial_{x}^{-1}\sigma_{i}-\frac{1}{2}\left(\sigma_{i}\right)_{x} & \frac{1}{2}\sigma_{i}
\end{pmatrix}\lambda^{n-i-1},\qquad n=-1,0,1\ldots,
\]
so that 
\[
V_{-1}=\begin{pmatrix}0 & 0\\
0 & 0
\end{pmatrix},\qquad V_{0}=\begin{pmatrix}-\frac{1}{4} & \frac{1}{2}x\\
\frac{1}{2}(\lambda-u)x & \frac{1}{4}
\end{pmatrix}.
\]

\subsection{Non-autonomous KdV hierarchy in the case $\mathcal{A}_{-1}$}

We now present the deformed KdV hierarchy $u_{t_{n}}=\mathbb{K}_{n}[u]$
for the case of the hereditary subalgebra $\mathcal{A}_{-1}$.

For the general initial conditions \eqref{wp2} the first members
\eqref{Rm1} of the non-autonomous KdV hierarchy take the form 
\begin{align*}
u_{t_{0}} & =1,\\
u_{t_{1}} & =1+\mathbf{a}_{1,1}(t_{1})\,u_{x},\\
u_{t_{2}} & =1+\Bigl[\mathbf{a}_{2,1}(t_{2})-\frac{3}{2}(t_{0}+t_{1})\,\mathbf{a}_{2,2}(t_{2})\Bigr]u_{x}+\mathbf{a}_{2,2}(t_{2})\Bigl[\frac{1}{4}u_{xxx}+\frac{3}{2}uu_{x}\Bigr],\\
u_{t_{3}} & =1+\left[\frac{3}{2}\,\partial_{t_{2}}^{-1}\mathbf{a}_{2,2}(t_{2})+\mathbf{a}_{3,1}(t_{3})-\frac{3}{2}(t_{0}+t_{1}+t_{2})\mathbf{a}_{3,2}(t_{3})+\frac{15}{8}(t_{0}+t_{1}+t_{2})^{2}\mathbf{a}_{3,3}(t_{3})\right]u_{x}\\
 & \quad+\Bigl[\mathbf{a}_{3,2}(t_{3})-\frac{5}{2}(t_{0}+t_{1}+t_{2})\mathbf{a}_{3,3}(t_{3})\Bigr]\Bigl[\frac{1}{4}u_{xxx}+\frac{3}{2}uu_{x}\Bigr]\\
 & \quad+\mathbf{a}_{3,3}(t_{3})\Bigl[\frac{1}{16}u_{5x}+\frac{5}{8}uu_{3x}+\frac{5}{4}u_{x}u_{xx}+\frac{15}{8}u^{2}u_{x}\Bigr].
\end{align*}

If we consider the initial conditions given by the choice $\mathbf{a}_{n,i}(t_{n})=\delta_{i,n}t_{n}$,
that is \eqref{inc1}, the deformed vector fields $\mathbb{K}_{n}$
are given by \eqref{Rm1} with \eqref{def1} specified by \eqref{def12}
with $\rho$ and $\kappa_{m}$ as in \eqref{kappa1}. In this case the first few non-autonomous vector
fields $\mathbb{K}_{n}$ are given by formulas in Example \ref{E2} and so the first few members of our hierarchy specify to
\begin{align*}
u_{t_{0}} & =1,\\
u_{t_{1}} & =1+t_{1}\,u_{x},\\
u_{t_{2}} & =1-\frac{3}{2}\,(t_{0}+t_{1})t_{2}\,u_{x}+t_{2}\Bigl[\frac{1}{4}u_{xxx}+\frac{3}{2}uu_{x}\Bigr],\\
u_{t_{3}} & =1+\left[\frac{3}{4}t_{2}^{2}+\frac{15}{8}(t_{0}+t_{1}+t_{2})^{2}t_{3}\right]u_{x}-\frac{5}{2}(t_{0}+t_{1}+t_{2})t_{3}\Bigl[\frac{1}{4}u_{xxx}+\frac{3}{2}uu_{x}\Bigr]\\
 & \quad+t_{3}\Bigl[\frac{1}{16}u_{5x}+\frac{5}{8}uu_{3x}+\frac{5}{4}u_{x}u_{xx}+\frac{15}{8}u^{2}u_{x}\Bigr].
\end{align*}

The Lax representation of the above non-autonomous KdV hierarchy
is given by \eqref{LaxKn}, with $L$ as in \eqref{LKdV} and $W_{n}$
defined by \eqref{Wnu} with $\varepsilon=-1,$ where, in the general case,
$\mathbf{u}_{n,i}$ are given by \eqref{def1} and in the case of special initial conditions
\eqref{inc1} the functions $\mathbf{u}_{n,i}$ are specified by \eqref{def12}.

\subsection{Non-autonomous KdV hierarchy in the case $\mathcal{A}_{0}$}

The deformed KdV hierarchy $u_{t_{n}}=\mathbb{K}_{n}[u]$ for the
case of the hereditary subalgebra $\mathcal{A}_{0}$ and with the general initial conditions \eqref{wp}
is given by the vector fields $\mathbb{K}_{n}$ as in \eqref{R0},
thus in particular 
\begin{align*}
u_{t_{0}} & =u+\frac{1}{2}xu_{x},\\
u_{t_{1}} & =u+\frac{1}{2}xu_{x}+\mathbf{a}_{1,1}(t_{1})e^{-\frac{1}{2}t_{0}}u_{x},\\
u_{t_{2}} & =u+\frac{1}{2}xu_{x}+\Bigl[\mathbf{c}_{1,1}(t_{1})e^{-\frac{1}{2}t_{0}}+\mathbf{a}_{2,1}(t_{2})e^{-\frac{1}{2}(t_{0}+t_{1})}\Bigr]u_{x}+\mathbf{a}_{2,2}(t_{2})e^{-\frac{3}{2}(t_{0}+t_{1})}\Bigl[\frac{1}{4}u_{xxx}+\frac{3}{2}uu_{x}\Bigr],\\
u_{t_{3}} & =u+\frac{1}{2}xu_{x}+\Bigl[\mathbf{c}_{1,1}(t_{1})e^{-\frac{1}{2}t_{0}}+\mathbf{c}_{2,1}(t_{2})e^{-\frac{1}{2}(t_{0}+t_{1})}+\mathbf{a}_{3,1}(t_{3})e^{-\frac{1}{2}(t_{0}+t_{1}+t_{2})}\Bigr]u_{x}\\
 & \quad+\Bigl[e^{-\frac{3}{2}(t_{0}+t_{1})}\mathbf{c}_{2,2}(t_{2})+\mathbf{a}_{3,2}(t_{3})e^{-\frac{3}{2}(t_{0}+t_{1}+t_{2})}\Bigr]\Bigl[\frac{1}{4}u_{xxx}+\frac{3}{2}uu_{x}\Bigr]\\
 & \quad+\mathbf{a}_{3,3}(t_{3})e^{-\frac{3}{2}(t_{0}+t_{1}+t_{2})}\Bigl[\frac{1}{64}u_{7x}+\frac{7}{32}uu_{5x}+\frac{21}{32}u_{x}u_{4x}+\frac{35}{32}u_{xx}u_{3x}+\frac{35}{32}u_{x}^{3}\\
 & \qquad\qquad\qquad\qquad\qquad\quad+\frac{35}{8}uu_{x}u_{xx}+\frac{35}{32}u^{2}u_{3x}+\frac{35}{16}u^{3}u_{x}\Bigr].
\end{align*}

If we consider the initial conditions \eqref{inc0} the deformed vector
fields $\mathbb{K}_{n}$ are given by \eqref{R0} with \eqref{def}
specified by \eqref{cni} with $\rho$ and $\kappa_{m}$ given by \eqref{kappa1}. In this case, as in
Example \ref{E0}, we have 
\begin{align*}
u_{t_{0}} & =u+\frac{1}{2}xu_{x},\\
u_{t_{1}} & =u+\frac{1}{2}xu_{x}+t_{1}e^{-\frac{1}{2}t_{0}}u_{x},\\
u_{t_{2}} & =u+\frac{1}{2}xu_{x}+\left(t_{1}+2e^{-\frac{1}{2}t_{1}}-2\right)e^{-\frac{1}{2}t_{0}}u_{x}+t_{2}e^{-\frac{3}{2}(t_{0}+t_{1})}\Bigl[\frac{1}{4}u_{xxx}+\frac{3}{2}uu_{x}\Bigr],\\
u_{t_{3}} & =u+\frac{1}{2}xu_{x}+\left(t_{1}+2e^{-\frac{1}{2}t_{1}}-2\right)e^{-\frac{1}{2}t_{0}}u_{x}+\left(t_{2}+\frac{2}{3}e^{-\frac{3}{2}t_{2}}-\frac{2}{3}\right)e^{-\frac{3}{2}(t_{0}+t_{1})}\Bigl[\frac{1}{4}u_{xxx}+\frac{3}{2}uu_{x}\Bigr]\\
 & \quad+t_{3}e^{-\frac{3}{2}(t_{0}+t_{1}+t_{2})}\Bigl[\frac{1}{64}u_{7x}+\frac{7}{32}uu_{5x}+\frac{21}{32}u_{x}u_{4x}+\frac{35}{32}u_{xx}u_{3x}+\frac{35}{32}u_{x}^{3}+\frac{35}{8}uu_{x}u_{xx}\\
 & \qquad\qquad\qquad\qquad\quad+\frac{35}{32}u^{2}u_{3x}+\frac{35}{16}u^{3}u_{x}\Bigr].
\end{align*}

The Lax representation of the above non-autonomous KdV hierarchies
is given by \eqref{LaxKn}, with $L$ given by \eqref{LKdV} and with respective $W_{n}$
defined by \eqref{Wnu} with $\varepsilon=0$.

\section{Non-autonomous DWW hierarchy}

\label{s6}

\subsection{Autonomous DWW hierarchy}

Let us now apply our theory to the DWW hierarchy. Its is a bi-Hamiltonian
hierarchy given by 
\begin{equation}
\begin{pmatrix}u\\
v
\end{pmatrix}_{t_{n}}=K_{n}=\pi_{0}\gamma_{n}=\pi_{1}\gamma_{n-1},\qquad n=1,2,...,\label{D}
\end{equation}
where $\gamma_{n}\equiv\delta H_{n}$ are exact one-forms and where
\[
\pi_{0}=\begin{pmatrix}-\frac{1}{2}v\partial_{x}-\frac{1}{2}\partial_{x}v & \partial_{x}\\
\partial_{x} & 0
\end{pmatrix},\qquad\pi_{1}=\begin{pmatrix}\frac{1}{4}\partial_{x}^{3}+\frac{1}{2}u\partial_{x}+\frac{1}{2}\partial_{x}u & 0\\
0 & \partial_{x}
\end{pmatrix}
\]
are two compatible Poisson operators. This hierarchy is generated
by \eqref{rK} and \eqref{1f} with the recursion operator and its
adjoint given by 
\[
N=\pi_{1}\pi_{0}^{-1}=\begin{pmatrix}0 & \frac{1}{4}\partial_{x}^{2}+u+\frac{1}{2}u_{x}\partial_{x}^{-1}\\
1 & v+\frac{1}{2}v_{x}\partial_{x}^{-1}
\end{pmatrix},\qquad N^{\dag}=\begin{pmatrix}0 & 1\\
\frac{1}{4}\partial_{x}^{2}+u-\frac{1}{2}\partial_{x}^{-1}u_{x} & v-\frac{1}{2}\partial_{x}^{-1}v_{x}
\end{pmatrix}.
\]
In particular, we have the symmetries 
\begin{align*}
K_{1} & =\begin{pmatrix}u_{x}\\
v_{x}
\end{pmatrix},\qquad K_{2}=\begin{pmatrix}\frac{1}{4}v_{xxx}+uv_{x}+\frac{1}{2}vu_{x}\\
u_{x}+\frac{3}{2}vv_{x}
\end{pmatrix},\\
K_{3} & =\begin{pmatrix}\frac{3}{8}v^{2}u_{x}+\frac{3}{2}uvv_{x}+\frac{3}{2}uu_{x}+\frac{1}{4}u_{3x}+\frac{3}{8}vv_{3x}+\frac{9}{8}v_{x}v_{2x}\\
\frac{3}{2}vu_{x}+\frac{3}{2}uv_{x}+\frac{15}{8}v^{2}v_{x}+\frac{1}{4}v_{3x}
\end{pmatrix},
\end{align*}
cosymmetries 
\[
\gamma_{0}=\begin{pmatrix}2\\
v
\end{pmatrix},\qquad\gamma_{1}=\begin{pmatrix}v\\
u+\frac{3}{4}v^{2}
\end{pmatrix},\qquad\gamma_{2}=\begin{pmatrix}u+\frac{3}{4}v^{2}\\
\frac{1}{4}v_{xx}+\frac{3}{2}uv+\frac{5}{8}v^{3}
\end{pmatrix},
\]
and functionals $H_{n}=\int h_{n}\,dx$, where 
\[
h_{0}=2u+\frac{1}{2}v^{2},\qquad h_{1}=uv+\frac{1}{4}v^{3},\qquad h_{2}=-\frac{1}{8}v_{x}^{2}+\frac{1}{2}u^{2}+\frac{3}{4}uv^{2}+\frac{5}{32}v^{4}.
\]

With the DWW hierarchy \eqref{D} one can also relate the hierarchy
of its master symmetries \eqref{rS} with the first few $\sigma_{n}$
of the form 
\[
\sigma_{-1}=\begin{pmatrix}-v\\
2
\end{pmatrix},\qquad\sigma_{0}=\begin{pmatrix}2u+xu_{x}\\
v+xv_{x}
\end{pmatrix},\qquad\sigma_{1}=\begin{pmatrix}\frac{3}{4}v_{xx}+\frac{1}{4}xv_{xxx}+uv+xuv_{x}+\frac{1}{2}xvu_{x}\\
xu_{x}+\frac{3}{2}xvv_{x}+v^{2}+2u
\end{pmatrix}.
\]
The master symmetries $\sigma_{n}$ are Hamiltonian as in \eqref{sH}
with 
\[
F=\int\left(2xu+\frac{1}{2}xv^{2}\right)\,dx.
\]
The symmetries $K_{i}$ and the master symmetries $\sigma_{j}$ constitute
the hereditary algebra \eqref{A} with 
\begin{equation}
\rho=1\qquad\text{so that}\qquad\kappa_{n}=n.\label{kappa2}
\end{equation}

The zero-curvature formulation \eqref{4.3a} for the DWW hierarchy
is given by 
\begin{equation}
L=\begin{pmatrix}0 & 1\\
\lambda^{2}-v\lambda-u & 0
\end{pmatrix}\label{LDWW}
\end{equation}
and by 
\begin{equation}
U_{n}=\frac{1}{2}\sum_{i=0}^{n-1}\begin{pmatrix}-\frac{1}{2}\left(\gamma_{i1}\right)_{x} & \gamma_{i1}\\
\left(\lambda^{2}-v\lambda-u\right)\gamma_{i1}-\frac{1}{2}\left(\gamma_{i1}\right)_{xx} & \frac{1}{2}\left(\gamma_{i1}\right)_{x}
\end{pmatrix}\lambda^{n-i-1},\qquad n=1,2,\ldots,\label{UnDWW}
\end{equation}
with $\gamma_{i1}$ denoting the first component of $\gamma_{i}$.
In particular, $U_{1}=L$ and 
\[
U_{2}=\begin{pmatrix}-\frac{1}{4}v_{x} & \lambda+\frac{1}{2}v\\
\lambda^{3}-\frac{1}{2}v\lambda^{2}-\left(u+\frac{1}{2}v^{2}\right)\lambda-\frac{1}{2}uv-\frac{1}{4}v_{2x} & \frac{1}{4}v_{x}
\end{pmatrix},
\]
\[
U_{3}=\begin{pmatrix}-\frac{1}{4}v_{x}\lambda-\frac{1}{4}u_{x}-\frac{3}{8}vv_{x} & \lambda^{2}+\frac{1}{2}v\lambda+\frac{3}{8}v^{2}+\frac{1}{2}u\\
(U_{3})_{21} & \frac{1}{4}v_{x}\lambda+\frac{1}{4}u_{x}+\frac{3}{8}vv_{x}
\end{pmatrix},
\]
where
{\small\[
(U_{3})_{21}=\lambda^{4}-\frac{1}{2}v\lambda^{3}-\frac{1}{8}\left(4u+v^{2}\right)\lambda^{2}-\frac{1}{8}\left(8uv+3v^{3}+2v_{2x}\right)\lambda-\frac{1}{8}\left(4u^{2}+3uv^{2}+2u_{2x}+3v_{x}^{2}+3vv_{2x}\right).
\]}

The deformed Lax formulation for the hierarchy of the DWW master symmetries
$\sigma_{n}$ is given by \eqref{4.6} with $V_{n}$ of the form 
\[
V_{n}=\frac{1}{2}\sum_{i=-1}^{n-1}\begin{pmatrix}-\frac{1}{2}\sigma_{i2} & \partial_{x}^{-1}\sigma_{i2}\\
\left(\lambda^{2}-v\lambda-u\right)\partial_{x}^{-1}\sigma_{i2}-\frac{1}{2}\left(\sigma_{i2}\right)_{x} & \frac{1}{2}\sigma_{i2}
\end{pmatrix}\lambda^{n-i-1},\qquad n=-1,0,\ldots\,,
\]
with $\sigma_{i2}$ denoting the second component of $\sigma_{i}$,
thus 
\begin{equation}
V_{-1}=0,\qquad V_{0}=\begin{pmatrix}-\frac{1}{2} & x\\
\left(\lambda^{2}-v\lambda-u\right)x & \frac{1}{2}
\end{pmatrix}.\label{VyDWW}
\end{equation}

\subsection{Non-autonomous DWW hierarchy in the case $\mathcal A_{-1}$}

The first few vector fields of the deformed DWW hierarchy $u_{t_{n}}=\mathbb{K}_{n}[u]$ in the
case $\mathcal A_{-1}$ and with the general initial conditions \eqref{wp2}
have the form 
\begin{align*}
\begin{pmatrix}u\\
v
\end{pmatrix}_{t_{0}} & =\begin{pmatrix}-v\\
2
\end{pmatrix},\\
\begin{pmatrix}u\\
v
\end{pmatrix}_{t_{1}} & =\begin{pmatrix}-v\\
2
\end{pmatrix}+\mathbf{a}_{1,1}(t_{1})\begin{pmatrix}u_{x}\\
v_{x}
\end{pmatrix},\\
\begin{pmatrix}u\\
v
\end{pmatrix}_{t_{2}} & =\begin{pmatrix}-v\\
2
\end{pmatrix}+\Bigl[\mathbf{a}_{2,1}(t_{2})-2(t_{0}+t_{1})\,\mathbf{a}_{2,2}(t_{2})\Bigr]\begin{pmatrix}u_{x}\\
v_{x}
\end{pmatrix}+\mathbf{a}_{2,2}(t_{2})\begin{pmatrix}\frac{1}{4}v_{xxx}+uv_{x}+\frac{1}{2}vu_{x}\\
u_{x}+\frac{3}{2}vv_{x}
\end{pmatrix}.
\end{align*}

If we consider the initial conditions \eqref{inc1} the deformed vector
fields $\mathbb{K}_{n}$ are given by \eqref{Rm1} with \eqref{def1}
specified by \eqref{def12} with $\rho$ and $\kappa_{n}$ as in \eqref{kappa2}. Explicitely,

\begin{align*}
\begin{pmatrix}u\\
v
\end{pmatrix}_{t_{0}} & =\begin{pmatrix}-v\\
2
\end{pmatrix},\\
\begin{pmatrix}u\\
v
\end{pmatrix}_{t_{1}} & =\begin{pmatrix}-v\\
2
\end{pmatrix}+t_{1}\begin{pmatrix}u_{x}\\
v_{x}
\end{pmatrix},\\
\begin{pmatrix}u\\
v
\end{pmatrix}_{t_{2}} & =\begin{pmatrix}-v\\
2
\end{pmatrix}-2(t_{0}+t_{1})t_{2}\begin{pmatrix}u_{x}\\
v_{x}
\end{pmatrix}+t_{2}\begin{pmatrix}\frac{1}{4}v_{xxx}+uv_{x}+\frac{1}{2}vu_{x}\\
u_{x}+\frac{3}{2}vv_{x}
\end{pmatrix}.
\end{align*}

\subsection{Non-autonomous DWW hierarchy in the case $\mathcal{A}_{0}$}

The deformed DWW hierarchy $u_{t_{n}}=\mathbb{K}_{n}[u]$ in the
case $\mathcal{A}_{0}$ and with the general initial conditions \eqref{wp}
is given by the vector fields \eqref{R0}, thus in particular 
\begin{align*}
\begin{pmatrix}u\\
v
\end{pmatrix}_{t_{0}} & =\begin{pmatrix}2u+xu_{x}\\
v+xv_{x}
\end{pmatrix},\\
\begin{pmatrix}u\\
v
\end{pmatrix}_{t_{1}} & =\begin{pmatrix}2u+xu_{x}\\
v+xv_{x}
\end{pmatrix}+\mathbf{a}_{1,1}(t_{1})e^{-t_{0}}\begin{pmatrix}u_{x}\\
v_{x}
\end{pmatrix},\\
\begin{pmatrix}u\\
v
\end{pmatrix}_{t_{2}} & =\begin{pmatrix}2u+xu_{x}\\
v+xv_{x}
\end{pmatrix}+\Bigl[\mathbf{c}_{1,1}(t_{1})e^{-t_{0}}+\mathbf{a}_{2,1}(t_{2})e^{-(t_{0}+t_{1})}\Bigr]\begin{pmatrix}u_{x}\\
v_{x}
\end{pmatrix}\\
 & \quad+\mathbf{a}_{2,2}(t_{2})e^{-2(t_{0}+t_{1})}\begin{pmatrix}\frac{1}{4}v_{xxx}+uv_{x}+\frac{1}{2}vu_{x}\\
u_{x}+\frac{3}{2}vv_{x}
\end{pmatrix}.
\end{align*}
For the initial conditions \eqref{inc0} from Example \ref{E0} the
deformed vector fields $\mathbb{K}_{n}$ are given by \eqref{R0}
with \eqref{def} specified by \eqref{cni} with $\rho$ and $\kappa_{n}$ as in \eqref{kappa2}. In
this case 
\begin{align*}
\begin{pmatrix}u\\
v
\end{pmatrix}_{t_{0}} & =\begin{pmatrix}2u+xu_{x}\\
v+xv_{x}
\end{pmatrix},\\
\begin{pmatrix}u\\
v
\end{pmatrix}_{t_{1}} & =\begin{pmatrix}2u+xu_{x}\\
v+xv_{x}
\end{pmatrix}+t_{1}e^{-t_{0}}\begin{pmatrix}u_{x}\\
v_{x}
\end{pmatrix},\\
\begin{pmatrix}u\\
v
\end{pmatrix}_{t_{2}} & =\begin{pmatrix}2u+xu_{x}\\
v+xv_{x}
\end{pmatrix}+\bigl(t_{1}+(e^{-t_{1}}-1)\bigr)e^{-t_{0}}\begin{pmatrix}u_{x}\\
v_{x}
\end{pmatrix}+t_{2}e^{-2(t_{0}+t_{1})}\begin{pmatrix}\frac{1}{4}v_{xxx}+uv_{x}+\frac{1}{2}vu_{x}\\
u_{x}+\frac{3}{2}vv_{x}
\end{pmatrix}.
\end{align*}

\bigskip{}
The Lax representation for all the above non-autonomous DWW hierarchies
is given by \eqref{LaxKn} with $L$ as in \eqref{LDWW} and with
$W_{n}$ as in \eqref{Wnu} with respective $\mathbf{u}_{n,i}$ and
$\varepsilon=-1$ or $0$.

\section{Non-autonomous AKNS hierarchy}

\label{s7}

\subsection{Autonomous AKNS hierarchy}

In the last example we apply our theory to the AKNS hierarchy. It
is a bi-Hamiltonian hierarchy given by 
\[
\begin{pmatrix}q\\
r
\end{pmatrix}_{t_{n}}=K_{n}=\pi_{0}\gamma_{n}=\pi_{1}\gamma_{n-1},\qquad n=1,2,...,
\]
where $\gamma_{n}=\delta H_{n}$ are exact one-forms and where 
\[
\pi_{0}=\begin{pmatrix}0 & i\\
-i & 0
\end{pmatrix},\qquad\pi_{1}=\begin{pmatrix}-q\partial_{x}^{-1}q & -\frac{1}{2}\partial_{x}+q\partial_{x}^{-1}r\\
-\frac{1}{2}\partial_{x}+r\partial_{x}^{-1}q & -r\partial_{x}^{-1}r
\end{pmatrix}
\]
are two compatible Poisson operators ($i^{2}=-1$ throughout this
whole section). This hierarchy is generated by \eqref{rK} and
\eqref{1f} with the recursion operator and its adjoint given by 
\[
N=\pi_{1}\pi_{0}^{-1}=i\begin{pmatrix}\frac{1}{2}\partial_{x}-q\partial_{x}^{-1}r & -q\partial_{x}^{-1}q\\
r\partial_{x}^{-1}r & -\frac{1}{2}\partial_{x}+r\partial_{x}^{-1}q
\end{pmatrix},\qquad N^{\dag}=-i\begin{pmatrix}\frac{1}{2}\partial_{x}-r\partial_{x}^{-1}q & r\partial_{x}^{-1}r\\
-q\partial_{x}^{-1}q & -\frac{1}{2}\partial_{x}+q\partial_{x}^{-1}r
\end{pmatrix}.
\]
In particular, we have the symmetries 
\[
K_{1}=i\begin{pmatrix}-2q\\
2r
\end{pmatrix},\qquad K_{2}=\begin{pmatrix}q_{x}\\
r_{x}
\end{pmatrix},\qquad K_{3}=i\begin{pmatrix}\frac{1}{2}q_{xx}-q^{2}r\\
-\frac{1}{2}r_{xx}+qr^{2}
\end{pmatrix},\qquad K_{4}=\begin{pmatrix}-\frac{1}{4}q_{xxx}+\frac{3}{2}qrq_{x}\\
-\frac{1}{4}r_{xxx}+\frac{3}{2}qrr_{x}
\end{pmatrix},
\]
cosymmetries 
\[
\gamma_{1}=\begin{pmatrix}-2r\\
-2q
\end{pmatrix},\qquad\gamma_{2}=i\begin{pmatrix}r_{x}\\
-q_{x}
\end{pmatrix},\qquad\gamma_{3}=\begin{pmatrix}\frac{1}{2}r_{xx}-qr^{2}\\
\frac{1}{2}q_{xx}-rq^{2}
\end{pmatrix},\qquad\gamma_{4}=i\begin{pmatrix}-\frac{1}{4}r_{xxx}+\frac{3}{2}qrr_{x}\\
\frac{1}{4}q_{xxx}-\frac{3}{2}qrq_{x}
\end{pmatrix},
\]
and respective Hamiltonians $H_{n}=\int h_{n}\,dx$, where
\[
h_{1}=-2rq,\qquad h_{2}=-i\,rq_{x},\qquad h_{3}=-\frac{1}{2}q_{x}r_{x}-\frac{1}{2}q^{2}r^{2},\qquad h_{4}=i\left(\frac{1}{4}rq_{xxx}+\frac{3}{4}q^{2}rr_{x}\right).
\]

With the AKNS hierarchy \eqref{D} one can relate the hierarchy of
its master symmetries \eqref{rS} with $\sigma_{-1}$ and $\sigma_{0}$
given by 
\[
\sigma_{-1}=i\begin{pmatrix}-2xq\\
2xr
\end{pmatrix},\qquad\sigma_{0}=\begin{pmatrix}(xq)_{x}\\
(xr)_{x}
\end{pmatrix}.
\]
One can directly check that in this case
\begin{equation}
\rho=0\qquad\text{and}\qquad\kappa_{n}=n-1 \label{kappa3}
\end{equation}
in the hereditary algebra \eqref{A}. Moreover, the master symmetries $\sigma_{n}$ are Hamiltonian as in
\eqref{sH} with 
\[
F=-2\int xqr\,dx.
\]
The zero-curvature formulation \eqref{4.3a} for the AKNS hierarchy
is given by \cite{book1,ML} 
\begin{equation}
L=\begin{pmatrix}-i\lambda & q\\
r & i\lambda
\end{pmatrix}\label{LAKNS}
\end{equation}
and by 
\begin{equation}
U_{n}=i\begin{pmatrix}-1 & 0\\
0 & 1
\end{pmatrix}\lambda^{n-1}+\frac{i}{2}\sum_{j=1}^{n-1}\begin{pmatrix}-\partial_{x}^{-1}\left[rK_{j1}+qK_{j2}\right] & K_{j1}\\
-K_{j2} & \partial_{x}^{-1}\left[rK_{j1}+qK_{j2}\right]
\end{pmatrix}\lambda^{n-j-1},\qquad n=1,2,\ldots\,.\text{ }\label{UnAKNS}
\end{equation}
In particular,
\begin{align*}
U_{1} & =\begin{pmatrix}-i & 0\\
0 & i
\end{pmatrix},\qquad U_{2}\equiv L,\qquad U_{2}=\begin{pmatrix}-i\lambda^{2}-\frac{i}{2}qr & q\lambda+\frac{i}{2}q_{x}\\
r\lambda-\frac{i}{2}r_{x} & i\lambda^{2}+\frac{i}{2}qr
\end{pmatrix},\\
U_{3} & =\begin{pmatrix}-i\lambda^{3}-\frac{i}{2}qr\lambda+\frac{1}{4}rq_{x}-\frac{1}{4}qr_{x} & q\lambda^{2}+\frac{i}{2}q_{x}\lambda-\frac{1}{4}q_{xx}+\frac{1}{2}q^{2}r\\
r\lambda^{2}-\frac{i}{2}r_{x}\lambda-\frac{1}{4}r_{xx}+\frac{1}{2}qr^{2} & i\lambda^{3}+\frac{i}{2}qr\lambda-\frac{1}{4}rq_{x}+\frac{1}{4}qr_{x}
\end{pmatrix}.
\end{align*}

The deformed Lax formulation for the hierarchy of the master symmetries
$\sigma_{n}$ is given by \eqref{4.6} with $V_{i}$ of the form 
\[
V_{n}=i\begin{pmatrix}-x & 0\\
0 & x
\end{pmatrix}\lambda^{n+1}+\frac{i}{2}\sum_{j=-1}^{n-1}\begin{pmatrix}-\partial_{x}^{-1}\left[r\sigma_{j1}+q\sigma_{j2}\right] & \sigma_{j1}\\
-\sigma_{j2} & \partial_{x}^{-1}\left[r\sigma_{j1}+q\sigma_{j2}\right]
\end{pmatrix}\lambda^{n-j-1},
\]
where $n=-1,0,1,...$, and thus 
\begin{equation}
V_{-1}=\begin{pmatrix}-ix & 0\\
0 & ix
\end{pmatrix},\qquad V_{0}=\begin{pmatrix}-ix\lambda & xq\\
xr & ix\lambda
\end{pmatrix}.\label{VyAKNS}
\end{equation}

\subsection{Non-autonomous AKNS hierarchy in the case $\mathcal{A}_{-1}$}

The first four systems of the deformed AKNS hierarchy $u_{t_{n}}=\mathbb{K}_{n}[u]$
in the case $\mathcal{A}_{-1}$ and with the general initial conditions
\eqref{wp2} are
\begin{align*}
\begin{pmatrix}q\\
r
\end{pmatrix}_{t_{0}} & =i\begin{pmatrix}-2xq\\
2xr
\end{pmatrix},\\
\begin{pmatrix}q\\
r
\end{pmatrix}_{t_{1}} & =i\begin{pmatrix}-2xq\\
2xr
\end{pmatrix}+i\mathbf{\,a}_{1,1}(t_{1})\,\begin{pmatrix}-2q\\
2r
\end{pmatrix},\\
\begin{pmatrix}q\\
r
\end{pmatrix}_{t_{2}} & =i\begin{pmatrix}-2xq\\
2xr
\end{pmatrix}+i\Bigl[\mathbf{a}_{2,1}(t_{2})-(t_{0}+t_{1})\,\mathbf{a}_{2,2}(t_{2})\Bigr]\begin{pmatrix}-2q\\
2r
\end{pmatrix}+\mathbf{a}_{2,2}(t_{2})\,\begin{pmatrix}q_{x}\\
r_{x}
\end{pmatrix},\\
\begin{pmatrix}q\\
r
\end{pmatrix}_{t_{3}} & =i\begin{pmatrix}-2xq\\
2xr
\end{pmatrix}+i\left[\partial_{t_{2}}^{-1}\mathbf{a}_{2,2}(t_{2})+\mathbf{a}_{3,1}(t_{3})-\tau_{2}\,\mathbf{a}_{3,2}(t_{3})+\tau_{2}^{2}\,\mathbf{a}_{3,3}(t_{3})\right]\begin{pmatrix}-2q\\
2r
\end{pmatrix}\\
 & \quad+\Bigl[\mathbf{a}_{3,2}(t_{3})-2\tau_{2}\,\mathbf{a}_{3,3}(t_{3})\Bigr]\begin{pmatrix}q_{x}\\
r_{x}
\end{pmatrix}+i\mathbf{\,a}_{3,3}(t_{3})\,\begin{pmatrix}\frac{1}{2}q_{xx}-q^{2}r\\
-\frac{1}{2}r_{xx}+qr^{2}
\end{pmatrix},
\end{align*}
where $\tau_{2}=t_{0}+t_{1}+t_{2}$.

If we consider the initial conditions \eqref{inc1} the deformed vector
fields $\mathbb{K}_{n}$ are given by \eqref{Rm1} with \eqref{def1}
specified by \eqref{def12} with $\rho$ and $\kappa_{n}$ as in \eqref{kappa3}, yielding
\begin{align*}
\begin{pmatrix}q\\
r
\end{pmatrix}_{t_{0}} & =i\begin{pmatrix}-2xq\\
2xr
\end{pmatrix},\\
\begin{pmatrix}q\\
r
\end{pmatrix}_{t_{1}} & =i\begin{pmatrix}-2xq\\
2xr
\end{pmatrix}+i\,t_{1}\,\begin{pmatrix}-2q\\
2r
\end{pmatrix},\\
\begin{pmatrix}q\\
r
\end{pmatrix}_{t_{2}} & =i\begin{pmatrix}-2xq\\
2xr
\end{pmatrix}-i(t_{0}+t_{1})\,t_{2}\,\begin{pmatrix}-2q\\
2r
\end{pmatrix}+t_{2}\,\begin{pmatrix}q_{x}\\
r_{x}
\end{pmatrix},\\
\begin{pmatrix}q\\
r
\end{pmatrix}_{t_{3}} & =i\begin{pmatrix}-2xq\\
2xr
\end{pmatrix}+i\left[\frac{1}{2}t_{2}^{2}+(t_{0}+t_{1}+t_{2})^{2}\,t_{3}\right]\begin{pmatrix}-2q\\
2r
\end{pmatrix}-2(t_{0}+t_{1}+t_{2})\,t_{3}\begin{pmatrix}q_{x}\\
r_{x}
\end{pmatrix}\\
 & \quad+i\,t_{3}\,\begin{pmatrix}\frac{1}{2}q_{xx}-q^{2}r\\
-\frac{1}{2}r_{xx}+qr^{2}
\end{pmatrix}.
\end{align*}

\subsection{Non-autonomous AKNS hierarchy in the case $\mathcal{A}_{0}$}

The deformed AKNS hierarchy $u_{t_{n}}=\mathbb{K}_{n}[u]$ in the
case $\mathcal{A}_{0}$ and with the general initial conditions \eqref{wp}
is given by the vector fields \eqref{R0}; the first few of them read
\begin{align*}
\begin{pmatrix}q\\
r
\end{pmatrix}_{t_{0}} & =\begin{pmatrix}(xq)_{x}\\
(xr)_{x}
\end{pmatrix},\\
\begin{pmatrix}q\\
r
\end{pmatrix}_{t_{1}} & =\begin{pmatrix}(xq)_{x}\\
(xr)_{x}
\end{pmatrix}+i\,\mathbf{a}_{1,1}(t_{1})\begin{pmatrix}-2q\\
2r
\end{pmatrix},\\
\begin{pmatrix}q\\
r
\end{pmatrix}_{t_{2}} & =\begin{pmatrix}(xq)_{x}\\
(xr)_{x}
\end{pmatrix}+i\,\mathbf{a}_{2,1}(t_{2})\begin{pmatrix}-2q\\
2r
\end{pmatrix}+\mathbf{a}_{2,2}(t_{2})e^{-(t_{0}+t_{1})}\begin{pmatrix}q_{x}\\
r_{x}
\end{pmatrix},\\
\begin{pmatrix}q\\
r
\end{pmatrix}_{t_{3}} & =\begin{pmatrix}(xq)_{x}\\
(xr)_{x}
\end{pmatrix}+i\,\mathbf{a}_{3,1}(t_{3})\begin{pmatrix}-2q\\
2r
\end{pmatrix}\\
 & \quad+\Bigl[e^{-(t_{0}+t_{1})}\mathbf{c}_{2,2}(t_{2})+\mathbf{a}_{3,2}(t_{3})e^{-\tau_{2}}\Bigr]\begin{pmatrix}q_{x}\\
r_{x}
\end{pmatrix}+i\,\mathbf{a}_{3,3}(t_{3})e^{-2\tau_{2}}\begin{pmatrix}\frac{1}{2}q_{xx}-q^{2}r\\
-\frac{1}{2}r_{xx}+qr^{2}
\end{pmatrix},
\end{align*}
where $\tau_{2}=t_{0}+t_{1}+t_{2}$. Notice that, in this case $\kappa_{1}=0$
and thus all $\mathbf{c}_{r,1}=0$ due to \eqref{c2}.

For the initial conditions \eqref{inc0} the deformed vector fields
$\mathbb{K}_{n}$ are given by \eqref{R0} with \eqref{def} specified
by \eqref{cni} with $\rho$ and $\kappa_{n}$ as in \eqref{kappa3}, which yields
\begin{align*}
\begin{pmatrix}q\\
r
\end{pmatrix}_{t_{0}} & =\begin{pmatrix}(xq)_{x}\\
(xr)_{x}
\end{pmatrix},\\
\begin{pmatrix}q\\
r
\end{pmatrix}_{t_{1}} & =\begin{pmatrix}(xq)_{x}\\
(xr)_{x}
\end{pmatrix}+i\,t_{1}\begin{pmatrix}-2q\\
2r
\end{pmatrix},\\
\begin{pmatrix}q\\
r
\end{pmatrix}_{t_{2}} & =\begin{pmatrix}(xq)_{x}\\
(xr)_{x}
\end{pmatrix}+i\Bigl[\Bigl(t_{1}-\frac{1}{\kappa_{1}}\Bigr)e^{-\kappa_{1}t_{0}}+\frac{1}{\kappa_{1}}e^{-\kappa_{1}(t_{0}+t_{1})}\Bigr]\begin{pmatrix}-2q\\
2r
\end{pmatrix}+t_{2}e^{-\kappa_{2}(t_{0}+t_{1})}\begin{pmatrix}q_{x}\\
r_{x}
\end{pmatrix},\\
\begin{pmatrix}q\\
r
\end{pmatrix}_{t_{3}} & =\begin{pmatrix}(xq)_{x}\\
(xr)_{x}
\end{pmatrix}+i\Bigl[\Bigl(t_{1}-\frac{1}{\kappa_{1}}\Bigr)e^{-\kappa_{1}t_{0}}+\frac{1}{\kappa_{1}}e^{-\kappa_{1}(t_{0}+t_{1})}\Bigr]\begin{pmatrix}-2q\\
2r
\end{pmatrix},\\
 & \quad+\Bigl[\Bigl(t_{2}-\frac{1}{\kappa_{2}}\Bigr)e^{-\kappa_{2}(t_{0}+t_{1})}+\frac{1}{\kappa_{2}}e^{-\kappa_{2}(t_{0}+t_{1}+t_{2})}\Bigr]\begin{pmatrix}q_{x}\\
r_{x}
\end{pmatrix}+i\,t_{3}e^{-\kappa_{3}(t_{0}+t_{1}+t_{2})}\begin{pmatrix}\frac{1}{2}q_{xx}-q^{2}r\\
-\frac{1}{2}r_{xx}+qr^{2}
\end{pmatrix}.
\end{align*}

\bigskip{}
The Lax representation for all the above non-autonomous AKNS hierarchies
are given by \eqref{LaxKn} with $L$ as in \eqref{LDWW} and with
$W_{n}$ as in \eqref{Wnu} with respective $\mathbf{u}_{n,i}$ and
$\varepsilon=-1$ or $0$.

\appendix

\section{%
\mbox{%
}}

\label{appA}

In this appendix we prove Theorem \ref{lemacik}. Consider the following
linear integral operators 
\[
\xi_{i}=1+\partial_{t_{i}}^{-1}\ad_{\mathbb{K}_{i}},\qquad i=0,1,\ldots,n-1
\]
acting in $\mathcal{A}$ (notice that $\xi_{i}=\xi_{i}(t_{0},\ldots,t_{i})$
due to \eqref{w1}). Then, by \eqref{B}, 
\[
\xi_{i}^{-1}=1-\partial_{t_{i}}^{-1}\ad_{\mathbb{K}_{i}}\xi_{i}^{-1}.
\]
Therefore 
\begin{align}
\partial_{t_{i}}\xi_{i} & =\partial_{t_{i}}+\ad_{\mathbb{K}_{i}},\nonumber \\
\partial_{t_{i}}\xi_{i}^{-1} & =\partial_{t_{i}}-\ad_{\mathbb{K}_{i}}\xi_{i}^{-1},\label{2.5}
\end{align}
and 
\[
\partial_{t_{i}}\xi_{j}=\xi_{j}\partial_{t_{i}}\quad\text{for}\quad j<i,
\]
while 
\[
\partial_{t_{i}}\xi_{j}^{-1}=-\xi_{j}^{-1}\frac{\partial\xi_{j}}{\partial t_{i}}\xi_{j}^{-1}+\xi_{j}^{-1}\partial_{t_{i}}=-\xi_{j}^{-1}[\xi_{j},\ad_{\mathbb{K}_{i}}]\xi_{j}^{-1}+\xi_{j}^{-1}\partial_{t_{i}}\quad\text{for}\quad j>i.
\]
Note that $\partial_{t_{i}}\xi_{j}$ and $\frac{\partial\xi_{j}}{\partial t_{i}}$
are two different operators, as $\partial_{t_{i}}\xi_{j}=\frac{\partial\xi_{j}}{\partial t_{i}}+\xi_{j}\partial_{t_{i}}$.
The formula above is due to the fact that for any time-dependent operator
$\xi_{j}$ 
\begin{equation}
\frac{\partial\xi_{j}^{-1}}{\partial t_{i}}=-\xi_{j}^{-1}\frac{\partial\xi_{j}}{\partial t_{i}}\xi_{j}^{-1}\label{Om1}
\end{equation}
and the fact that the explicit time dependence of $\xi_{j}$ on $t_{i}$
is 
\begin{align}
\frac{\partial\xi_{j}}{\partial t_{i}} & =\partial_{t_{j}}^{-1}\frac{\partial\ad_{\mathbb{K}_{j}}}{\partial t_{i}}\overset{\eqref{w1}}{=}-\partial_{t_{j}}^{-1}\ad_{[\mathbb{K}_{i},\mathbb{K}_{j}]}=-\partial_{t_{j}}^{-1}[\ad_{\mathbb{K}_{i}},\ad_{\mathbb{K}_{j}}]=[\partial_{t_{j}}^{-1}\ad_{\mathbb{K}_{j}},\ad_{\mathbb{K}_{i}}]\label{3}\\
 & =[\xi_{j},\ad_{\mathbb{K}_{i}}]\quad\text{for}\quad i<j.\nonumber 
\end{align}
Notice that $\frac{\partial\xi_{j}}{\partial t_{i}}=0$ for $i>j$.

Denote now $\bm{t}=(t_{0},\ldots,t_{n-1})$ and consider \eqref{2},
that is 
\[
\mathbb{K}\left(\bm{t}\right)=\xi_{n-1}^{-1}\cdots\xi_{1}^{-1}\xi_{0}^{-1}\bar{\mathbb{K}}.
\]
We will now show that $\mathbb{K}\left(\bm{t}\right)$ is the solution
of the IVP \eqref{1}. Naturally, $\mathbb{K}\left(\bm{0}\right)=\bar{\mathbb{K}}$.
Using \eqref{Om1} we obtain 
\begin{align*}
\frac{\partial\mathbb{K}(\bm{t})}{\partial t_{i}} & =\partial_{t_{i}}\xi_{n-1}^{-1}\cdots\xi_{1}^{-1}\xi_{0}^{-1}\bar{\mathbb{K}}\\
 & =-\xi_{n-1}^{-1}\frac{\partial\xi_{n-1}}{\partial t_{i}}\xi_{n-1}^{-1}\xi_{n-2}^{-1}\cdots\xi_{1}^{-1}\xi_{0}^{-1}\bar{\mathbb{K}}+\xi_{n-1}^{-1}\partial_{t_{i}}\xi_{n-2}^{-1}\cdots\xi_{1}^{-1}\xi_{0}^{-1}\bar{\mathbb{K}}=\cdots=\\
 & =-\sum_{j=i+1}^{n-1}\xi_{n-1}^{-1}\cdots\xi_{j}^{-1}\frac{\partial\xi_{j}}{\partial t_{i}}\xi_{j}^{-1}\cdots\xi_{i+1}^{-1}\xi_{i}^{-1}\cdots\xi_{0}^{-1}\bar{\mathbb{K}}+\xi_{n-1}^{-1}\cdots\xi_{i+1}^{-1}\partial_{t_{i}}\xi_{i}^{-1}\cdots\xi_{0}^{-1}\bar{\mathbb{K}}
\end{align*}
and further by \eqref{3} and \eqref{2.5} 
\begin{align*}
\frac{\partial\mathbb{K}(\bm{t})}{\partial t_{i}} & =-\sum_{j=i+1}^{n-1}\xi_{n-1}^{-1}\cdots\xi_{j}^{-1}[\xi_{j},\ad_{\mathbb{K}_{i}}]\xi_{j}^{-1}\cdots\xi_{i+1}^{-1}\xi_{i}^{-1}\cdots\xi_{0}^{-1}\bar{\mathbb{K}}-\xi_{n-1}^{-1}\cdots\xi_{i+1}^{-1}\ad_{\mathbb{K}_{i}}\xi_{i}^{-1}\cdots\xi_{0}^{-1}\bar{\mathbb{K}}\\
 & =-\ad_{\mathbb{K}_{i}}\xi_{n-1}^{-1}\cdots\xi_{i+1}^{-1}\xi_{i}^{-1}\cdots\xi_{0}^{-1}\bar{\mathbb{K}}-\sum_{j=i+1}^{n-2}\xi_{n-1}^{-1}\cdots\xi_{j+1}^{-1}\ad_{\mathbb{K}_{i}}\xi_{j}^{-1}\cdots\xi_{i+1}^{-1}\xi_{i}^{-1}\cdots\xi_{0}^{-1}\bar{\mathbb{K}}\\
 & \quad+\sum_{j=i+1}^{n-1}\xi_{n-1}^{-1}\cdots\xi_{j}^{-1}\ad_{\mathbb{K}_{i}}\xi_{j-1}^{-1}\cdots\xi_{i}^{-1}\xi_{i-1}^{-1}\cdots\xi_{0}^{-1}\bar{\mathbb{K}}-\xi_{n-1}^{-1}\cdots\xi_{i+1}^{-1}\ad_{\mathbb{K}_{i}}\xi_{i}^{-1}\cdots\xi_{0}^{-1}\bar{\mathbb{K}}\\
 & =-\ad_{\mathbb{K}_{i}}\xi_{n-1}^{-1}\cdots\xi_{0}^{-1}\bar{\mathbb{K}}\equiv-\ad_{\mathbb{K}_{i}}\mathbb{K}(\bm{t}).
\end{align*}

\section{%
\mbox{%
}}

\label{appB}

We prove here \eqref{wki2} and \eqref{ab} for the cases $\varepsilon=-1$
and $\varepsilon=0$, respectively. In either case we have 
\[
[\sigma_{\varepsilon},K_{i}]=\kappa_{i}K_{i+\varepsilon},\qquad i\geqslant1,\quad K_{0}\equiv0.
\]
Consider \eqref{R0} and \eqref{Rm1}, thus 
\[
\mathbb{K}_{n}=\sigma_{\varepsilon}+\sum_{i=1}^{n}\mathbf{u}_{n,i}(t_{0},\ldots,t_{n})K_{i},\qquad\varepsilon=0,-1.
\]
Then for $m<n$: 
\begin{align*}
 & \frac{\partial\mathbb{K}_{n}}{\partial t_{m}}+\left[\mathbb{K}_{m},\mathbb{K}_{n}\right]=\sum_{i=1}^{n}(\mathbf{u}_{n,i})_{t_{m}}K_{i}+\sum_{i=1}^{n}\mathbf{u}_{n,i}\left[\sigma_{\varepsilon},K_{i}\right]+\sum_{j=1}^{m}\mathbf{u}_{m,j}\left[K_{j},\sigma_{\varepsilon}\right]\\
 & \qquad=\sum_{i=1}^{n}(\mathbf{u}_{n,i})_{t_{m}}K_{i}+\sum_{i=1-\varepsilon}^{n}\kappa_{i}\mathbf{u}_{n,i}K_{i+\varepsilon}-\sum_{i=1-\varepsilon}^{m}\kappa_{i}\mathbf{u}_{m,i}K_{i+\varepsilon}\\
 & \qquad=-\varepsilon(\mathbf{u}_{n,n})_{t_{m}}K_{n}+\sum_{i=m+1}^{n}\left[(\mathbf{u}_{n,i+\varepsilon})_{t_{m}}+\kappa_{i}\mathbf{u}_{n,i}\right]K_{i+\varepsilon}\\
 & \qquad\quad+\sum_{i=1-\varepsilon}^{m}\left[(\mathbf{u}_{n,i+\varepsilon})_{t_{m}}+\kappa_{i}(\mathbf{u}_{n,i}-\mathbf{u}_{m,i})\right]K_{i+\varepsilon}.
\end{align*}
Thus, $\frac{\partial\mathbb{K}_{n}}{\partial t_{m}}+\left[\mathbb{K}_{m},\mathbb{K}_{n}\right]=0$
for $m<n$ if and only if 
\[
\begin{cases}
\varepsilon(\mathbf{u}_{n,n})_{t_{m}}=0,\\
(\mathbf{u}_{n,i})_{t_{m}}+\kappa_{i-\varepsilon}\mathbf{u}_{n,i-\varepsilon}=0, & m+1+\varepsilon\leqslant i\leqslant n+\varepsilon,\\
(\mathbf{u}_{n,i})_{t_{m}}+\kappa_{i-\varepsilon}\left(\mathbf{u}_{n,i-\varepsilon}-\mathbf{u}_{m,i-\varepsilon}\right)=0, & 1\leqslant i\leqslant m+\varepsilon,
\end{cases}
\]
which yields \eqref{ab} for $\varepsilon=0$ and \eqref{wki2} for
$\varepsilon=-1$.

\section{%
\mbox{%
}}

\label{appC}

We prove here \eqref{cntj}. Differentiating \eqref{cn} with respect to $t_{j}$
yields 
\begin{align*}
(\mathbf{c}_{n})_{t_{j}} & =\sum_{m=j+1}^{n-1}\sum_{r=3}^{m}\sum_{s=1}^{r-2}\frac{(-1)^{r-1}[\kappa_{r}]!}{(r-s-2)!}(\tau_{m-1})^{r-s-2}\left(\partial_{t_{m}}^{-1}\right)^{s}\mathbf{a}_{m,r}(t_{m})\\
 & \quad+\sum_{r=2}^{j}\sum_{s=1}^{r-1}\frac{(-1)^{r-1}[\kappa_{r}]!}{(r-s-1)!}(\tau_{j-1})^{r-s-1}\left(\partial_{t_{j}}^{-1}\right)^{s-1}\mathbf{a}_{j,r}(t_{j})\\
 & \quad+\sum_{r=2}^{n}\frac{(-1)^{r}[\kappa_{r}]!}{(r-2)!}(\tau_{n-1})^{r-2}\mathbf{a}_{n,r}(t_{n}),
\end{align*}
and further 
\begin{align*}
(\mathbf{c}_{n})_{t_{j}} & =\sum_{m=j}^{n-1}\sum_{r=3}^{m}\sum_{s=1}^{r-2}\frac{(-1)^{r-1}[\kappa_{r}]!}{(r-s-2)!}(\tau_{m-1})^{r-s-2}\left(\partial_{t_{m}}^{-1}\right)^{s}\mathbf{a}_{m,r}(t_{m})\\
 & \quad-\sum_{r=3}^{j}\sum_{s=1}^{r-2}\frac{(-1)^{r-1}[\kappa_{r}]!}{(r-s-2)!}(\tau_{j-1})^{r-s-2}\left(\partial_{t_{j}}^{-1}\right)^{s}\mathbf{a}_{j,r}(t_{j})\\
 & \quad+\sum_{r=2}^{j}\sum_{s=1}^{r-1}\frac{(-1)^{r-1}[\kappa_{r}]!}{(r-s-1)!}(\tau_{j-1})^{r-s-1}\left(\partial_{t_{j}}^{-1}\right)^{s-1}\mathbf{a}_{j,r}(t_{j})\\
 & \quad+\sum_{r=2}^{n}\frac{(-1)^{r}[\kappa_{r}]!}{(r-2)!}(\tau_{n-1})^{r-2}\mathbf{a}_{n,r}(t_{n}).
\end{align*}
Thus 
\begin{align*}
(\mathbf{c}_{n})_{t_{j}} & =\sum_{m=j}^{n-1}\sum_{r=3}^{m}\sum_{s=1}^{r-2}\frac{(-1)^{r-1}[\kappa_{r}]!}{(r-s-2)!}(\tau_{m-1})^{r-s-2}\left(\partial_{t_{m}}^{-1}\right)^{s}\mathbf{a}_{m,r}(t_{m})\\
 & \quad+\left(\sum_{r=2}^{j}\sum_{s=0}^{r-2}-\sum_{r=3}^{j}\sum_{s=1}^{r-2}\right)\frac{(-1)^{r-1}[\kappa_{r}]!}{(r-s-2)!}(\tau_{j-1})^{r-s-2}\left(\partial_{t_{j}}^{-1}\right)^{s}\mathbf{a}_{j,r}(t_{j})\\
 & \quad+\sum_{r=2}^{n}\frac{(-1)^{r}[\kappa_{r}]!}{(r-2)!}(\tau_{n-1})^{r-2}\mathbf{a}_{n,r}(t_{n}),
\end{align*}
and finally 
\begin{align*}
(\mathbf{c}_{n})_{t_{j}} & =\left(\sum_{m=2}^{n-1}-\sum_{m=2}^{j-1}\right)\sum_{r=3}^{m}\sum_{s=1}^{r-2}\frac{(-1)^{r-1}[\kappa_{r}]!}{(r-s-2)!}(\tau_{m-1})^{r-s-2}\left(\partial_{t_{m}}^{-1}\right)^{s}\mathbf{a}_{m,r}(t_{m})\\
 & \quad+\left(\sum_{r=2}^{n}-\sum_{r=2}^{j}\right)\frac{(-1)^{r}[\kappa_{r}]!}{(r-2)!}(\tau_{n-1})^{r-2}\mathbf{a}_{n,r}(t_{n})\\
 & \equiv(\mathbf{c}_{n})_{t_{0}}-(\mathbf{c}_{j})_{t_{0}}.
\end{align*}

\section{%
\mbox{%
}}

\label{appD}

We prove here Theorem \ref{vu}. We start by observing that, due to
\eqref{A} 
\[
(\Psi_{s_{i}})_{\tau_{j}}-(\Psi_{\tau_{j}})_{s_{i}}=\Psi^{\prime}\bigl[[\sigma_{j},K_{i}]\bigr]=\kappa_{i}\Psi^{\prime}[K_{i+j}]=\kappa_{i}\Psi_{s_{i+j}},
\]
which implies 
\[
(U_{i})_{\tau_{j}}-(V_{j})_{s_{i}}+[U_{i},V_{j}]+\lambda^{j+1}(U_{i})_{\lambda}=\kappa_{i}U_{i+j},
\]
or, equivalently, 
\[
U_{i}^{\prime}[\sigma_{j}]-V_{j}^{\prime}[K_{i}]+[U_{i},V_{j}]+\lambda^{j+1}(U_{i})_{\lambda}=\kappa_{i}U_{i+j},
\]
so that, in particular 
\begin{equation}
U_{i}^{\prime}[\sigma_{\varepsilon}]-V_{\varepsilon}^{\prime}[K_{i}]+[U_{i},V_{\varepsilon}]+\lambda^{\varepsilon+1}(U_{i})_{\lambda}=\kappa_{i}U_{i+\varepsilon},\qquad\varepsilon=-1\text{ or }0.\label{ZCUV}
\end{equation}
Obviously 
\[
\frac{\partial W_{n}}{\partial t_{m}}=\sum_{i=1}^{n}(\mathbf{v}_{n,i})_{t_{m}}U_{i}.
\]
Further 
\begin{align*}
 & W_{n}^{\prime}[\mathbb{K}_{m}]-W_{m}^{\prime}[\mathbb{K}_{n}]\overset{\eqref{Wn}}{=}V_{\varepsilon}^{\prime}[\mathbb{K}_{m}]+\sum_{i=1}^{n}\mathbf{v}_{n,i}U_{i}^{\prime}[\mathbb{K}_{m}]-V_{\varepsilon}^{\prime}[\mathbb{K}_{n}]-\sum_{j=1}^{m}\mathbf{v}_{m,j}U_{j}^{\prime}[\mathbb{K}_{n}]\\
 & \overset{\eqref{hd}}{=}\sum_{j=1}^{m}\mathbf{v}_{m,j}\left(V_{\varepsilon}^{\prime}[K_{j}]-U_{j}^{\prime}[\sigma_{\varepsilon}]\right)-\sum_{i=1}^{n}\mathbf{v}_{n,i}\left(V_{\varepsilon}^{\prime}[K_{i}]-U_{i}^{\prime}[\sigma_{\varepsilon}]\right)+\sum_{i=1}^{n}\sum_{j=1}^{m}\mathbf{v}_{n,i}\mathbf{v}_{m,j}\left(U_{i}^{\prime}[K_{j}]-U_{j}^{\prime}[K_{i}]\right)\\
 & \overset{\eqref{ZCUV},\eqref{ZCU}}{=}\sum_{j=1}^{m}\mathbf{v}_{m,j}\left([U_{j},V_{\varepsilon}]+\lambda^{\varepsilon+1}(U_{j})_{\lambda}-\kappa_{j}U_{j+\varepsilon}\right)\\
 & \qquad-\sum_{i=1}^{n}\mathbf{v}_{n,i}\left([U_{i},V_{\varepsilon}]+\lambda^{\varepsilon+1}(U_{i})_{\lambda}-\kappa_{i}U_{i+\varepsilon}\right)-\sum_{i=1}^{n}\sum_{j=1}^{m}\mathbf{v}_{n,i}\mathbf{v}_{m,j}[U_{i},U_{j}],
\end{align*}
and 
\[
[W_{n},W_{m}]\overset{\eqref{Wn}}{=}\sum_{j=1}^{m}\mathbf{v}_{m,j}[V_{\varepsilon},U_{j}]+\sum_{i=1}^{n}\mathbf{v}_{n,i}[U_{i},V_{\varepsilon}]+\sum_{i=1}^{n}\sum_{j=1}^{m}\mathbf{v}_{n,i}\mathbf{v}_{m,j}[U_{i},U_{j}],
\]
and finally 
\[
\lambda^{\varepsilon+1}(W_{n})_{\lambda}-\lambda^{\varepsilon+1}(W_{m})_{\lambda}\overset{\eqref{Wn}}{=}\lambda^{\varepsilon+1}\sum_{i=1}^{n}\mathbf{v}_{n,i}(U_{i})_{\lambda}-\lambda^{\varepsilon+1}\sum_{j=1}^{m}\mathbf{v}_{m,j}(U_{j})_{\lambda}.
\]
Thus, for $m<n$ 
\begin{align*}
 & \frac{\partial W_{n}}{\partial t_{m}}+W_{n}^{\prime}[\mathbb{K}_{m}]-W_{m}^{\prime}[\mathbb{K}_{n}]+[W_{n},W_{m}]+\lambda^{\varepsilon+1}(W_{n})_{\lambda}-\lambda^{\varepsilon+1}(W_{m})_{\lambda}=\\
 & \qquad=\sum_{i=1}^{n}(\mathbf{v}_{n,i})_{t_{m}}U_{i}-\sum_{j=1-\varepsilon}^{m}\kappa_{i}\mathbf{v}_{m,i}U_{i+\varepsilon}+\sum_{i=1-\varepsilon}^{n}\kappa_{i}\mathbf{v}_{n,i}U_{i+\varepsilon}\\
 & \qquad=-\varepsilon(\mathbf{v}_{n,n})_{t_{m}}U_{n}+\sum_{i=m+1}^{n}\left[(\mathbf{v}_{n,i+\varepsilon})_{t_{m}}+\kappa_{i}\mathbf{v}_{n,i}\right]U_{i+\varepsilon}\\
 & \qquad\quad+\sum_{i=1-\varepsilon}^{m}\left[(\mathbf{v}_{n,i+\varepsilon})_{t_{m}}+\kappa_{i}(\mathbf{v}_{n,i}-\mathbf{v}_{m,i})\right]U_{i+\varepsilon}=0.
\end{align*}
Thus, \eqref{ZC1} holds if and only if 
\[
\begin{cases}
\varepsilon(\mathbf{v}_{n,n})_{t_{m}}=0,\\
(\mathbf{v}_{n,i})_{t_{m}}+\kappa_{i-\varepsilon}\mathbf{v}_{n,i-\varepsilon}=0, & m+1+\varepsilon\leqslant i\leqslant n+\varepsilon\\
(\mathbf{v}_{n,i})_{t_{m}}+\kappa_{i-\varepsilon}\left(\mathbf{v}_{n,i-\varepsilon}-\mathbf{v}_{m,i-\varepsilon}\right)=0, & 1\leqslant i\leqslant m+\varepsilon
\end{cases}
\]
which coincides with equations \eqref{wki2} for $\varepsilon=-1$ and
with \eqref{ab} for $\varepsilon=0$. Thus, in both cases, $\mathbf{v}_{n,i}\equiv\mathbf{u}_{n,i}$
for all $n,i$.

\end{document}